\documentclass[11pt]{article}
\usepackage{amsmath,amssymb,amsfonts,amsthm,epsfig,verbatim}
\usepackage{times}
\usepackage{color}
\newif\ifhyper\IfFileExists{hyperref.sty}{\hypertrue}{\hyperfalse}
\hypertrue
\ifhyper\usepackage{hyperref}\fi

\newtheorem{theorem}{Theorem}

\newtheorem{lemma}[theorem]{Lemma}
\newtheorem{proposition}[theorem]{Proposition}

\newtheorem{claim}[theorem]{Claim}
\newtheorem{fact}[theorem]{Fact}
\newtheorem{obs}[theorem]{Observation}
\newtheorem{definition}[theorem]{Definition}

\newcommand{\eps}{\epsilon}
\newcommand{\pr}{\mathbf{Pr}}

\newcommand{\etal}{{\em et al.\ }}

\newcommand{\equ}[1]{

\begin{equation}
#1
\end{equation}}
\newcommand{\equn}[1]{$$ #1 $$}

\newcommand{\sgn}{\mathrm{sign}}
\newcommand{\sign}{\mathrm{sign}}
\renewcommand{\span}{\mathrm{span}}

\newcommand{\ignore}[1]{}

\newcommand{\cref}[1]{Corollary~\ref{cor:#1}}

\newcommand{\bits}{\{-1,1\}}
\newcommand{\bn}{\bits^n}

\newcommand{\R}{{\mathbb{R}}}

\newcommand{\Z}{{\mathbb Z}}
\newcommand{\E}{\operatorname{{\bf E}}}

\newcommand{\littlesum}{\mathop{\textstyle \sum}}

\newcommand{\poly}{\mathrm{poly}}
\newcommand{\quasipoly}{\mathrm{quasipoly}}

\newcommand{\wh}{\widehat}
\newcommand{\eqdef}{\stackrel{\textrm{def}}{=}}

\newcommand{\mychows}{\vec{\chi}}
\renewcommand{\Pr}{\operatorname{{\bf Pr}}}
\newcommand{\dist}{\mathrm{dist}}
\newcommand{\dchow}{d_{\mathrm{Chow}}}

\newcommand{\h}{\mathbf{H}}
\newcommand{\bA}{\mathbf{A}}
\newcommand{\mcwa}{{\mathcal W}}

\newcommand{\opt}{\mathsf{opt}}

\newcommand{\newad}[1]{{#1}}

\newcommand{\chowallow}{\kappa}

\newcommand{\alg}{\mathcal{A}}

\makeatletter
\renewcommand{\subsection}{\@startsection{subsection}{2}{0pt}{-6pt}{-5pt}{\normalsize\bf}}
\renewcommand{\subsubsection}{\@startsection{subsubsection}{3}{0pt}{-12pt}{-5pt}{\normalsize\bf}}
\makeatother

\newcommand{\rnote}[1]{\footnote{{\bf [[Rocco: {#1}\bf ]] }}}

\newcommand{\anote}[1]{\footnote{{\bf [[Anindya: {#1}\bf ]] }}}

\date{}

\begin{document}

\title{
Nearly optimal solutions for the Chow Parameters Problem and low-weight
approximation of halfspaces
%Regularity, Chow Parameters and Near-Optimal Integer Weight Approximation of
%Halfspaces
% or ``A Quasi-FPTAS for the Chow Parameters Problem''
}

\author{Anindya De\thanks{{\tt anindya@cs.berkeley.edu}.  Research supported by NSF award CCF-0915929, CCF-1017403 and CCF-1118083.}\\
University of California, Berkeley\\
\and
Ilias Diakonikolas\thanks{{\tt ilias@cs.berkeley.edu}.  Research supported by a Simons Postdoctoral Fellowship.}\\
University of California, Berkeley\\
\and
Vitaly Feldman\thanks{{\tt vitaly@post.harvard.edu.}} \\
IBM Almaden Research Center\\
\and Rocco A.\ Servedio\thanks{{\tt rocco@cs.columbia.edu}. Supported by NSF grants CNS-0716245, CCF-0915929, and CCF-1115703.}\\
Columbia University\\
}

\maketitle

\setcounter{page}{0}

\thispagestyle{empty}

~
\vskip -.5in
~

\begin{abstract}
The \emph{Chow parameters} of a Boolean function $f: \{-1,1\}^n \to \{-1,1\}$ are its $n+1$ degree-0 and
degree-1 Fourier coefficients.  It has been known since 1961 \cite{Chow:61, Tannenbaum:61} that the (exact values of the) Chow parameters of
any linear threshold function $f$ uniquely specify $f$ within the space of all Boolean functions, but until
recently \cite{OS11:chow} nothing was known about efficient algorithms for \emph{reconstructing} $f$ (exactly or approximately) from exact or approximate values of its Chow parameters.  We refer to this reconstruction problem as the \emph{Chow Parameters Problem.}

Our main result  is a new algorithm for the Chow Parameters Problem which,
given (sufficiently accurate approximations to) the Chow parameters of any linear threshold function $f$, runs in time
$\tilde{O}(n^2)\cdot (1/\eps)^{O(\log^2(1/\eps))}$ and
with high probability outputs a representation of an LTF  $f'$ that is $\eps$-close to $f$.
The only previous algorithm \cite{OS11:chow} had running time $\poly(n) \cdot 2^{2^{\tilde{O}(1/\eps^2)}}.$

As a byproduct of our approach, we show that for any linear threshold function $f$ over $\{-1,1\}^n$,
there is a linear threshold function $f'$ which is $\eps$-close to $f$ and has all weights that are integers at most $\sqrt{n} \cdot (1/\eps)^{O(\log^2(1/\eps))}$.
This significantly improves the best previous result of~\cite{DiakonikolasServedio:09} which gave
a $\poly(n) \cdot 2^{\tilde{O}(1/\eps^{2/3})}$ weight bound, and is close to the
known lower bound of
$\max\{\sqrt{n},$ $(1/\eps)^{\Omega(\log \log (1/\eps))}\}$ \cite{Goldberg:06b,Servedio:07cc}.
Our techniques also yield improved algorithms for related problems in learning theory.

In addition to being significantly stronger than previous work, our results
are obtained using conceptually simpler proofs.
The two main ingredients underlying our results are (1) a new structural
result showing that for $f$ any linear threshold function and $g$ any bounded
function, if the Chow parameters of $f$ are close to the Chow
parameters of $g$ then $f$ is close to $g$; (2) a new boosting-like algorithm
that given approximations to the Chow parameters of a linear threshold function outputs a bounded
function whose Chow parameters are close to those of $f$.
\end{abstract}

\newpage

\section{Introduction}\label{sec:intro}

\subsection{Background and motivation.} \label{sec:background}
A \emph{linear threshold function}, or LTF, over $\{-1,1\}^n$ is a Boolean function $f: \{-1,1\}^n \to \{-1,1\}$ of the form $$f(x) = \sgn\left( \littlesum_{i=1}^n w_ix_i - \theta \right),$$ where $w_1, \ldots, w_n, \theta \in \R$. The function $\sign(z)$ takes value $1$ if $z \geq 0$ and
takes value $-1$ if $z < 0$; the $w_i$'s are the \emph{weights} of $f$ and $\theta$ is the \emph{threshold}.
Linear threshold functions have been intensively studied for decades in many different fields.
They are variously known as  ``halfspaces'' or ``linear separators'' in machine learning and computational learning theory, ``Boolean threshold functions,'' ``(weighted) threshold gates'' and
``(Boolean) perceptrons (of order 1)'' in computational complexity, and as ``weighted majority games''
in voting theory and the theory of social choice.  Throughout this paper we shall refer to them simply as LTFs.

The \emph{Chow parameters} of a function $f: \{-1,1\}^n \to \R$ are the $n+1$
values
\[
\wh{f}(0)= \E[f(x)], \quad \wh{f}(i) = \E[f(x)x_i] \text{~for $i=1,\dots,n$},
\]
i.e. the $n+1$ degree-0 and degree-1 Fourier coefficients of $f$.  (Here and throughout the paper, all probabilities and expectations are with respect to the uniform distribution over
$\{-1,1\}^n$ unless otherwise indicated.)
It is easy to see that
in general the Chow parameters of a Boolean function may provide very little information about $f$; for example, any parity function on at least two variables has all its Chow parameters equal to 0.  However, in
a surprising result, C.-K.\ Chow \cite{Chow:61} showed that the Chow parameters of
an LTF $f$ \emph{uniquely} specify $f$ within the space of all Boolean functions mapping $\{-1,1\}^n \to \{-1,1\}.$
Chow's proof (given in Section~\ref{ssec:chow}) is simple and elegant, but is completely non-constructive; it does not give any clues as to how one might use the Chow parameters to find $f$ (or an LTF that is close to $f$).  This naturally gives rise to the following algorithmic question, which we refer to as the ``Chow Parameters Problem:''

\begin{quote}
{\bf The Chow Parameters Problem} (rough statement):  Given (exact or approximate) values for the Chow
parameters of an unknown LTF $f$, output an (exact or approximate) representation of $f$ as
$\sign(v_1 x_1 + \cdots + v_n x_n - \theta').$
\end{quote}

\noindent
{\bf Motivation and Prior Work.}  We briefly survey some previous research on the Chow Parameters problem (see Section~1.1 of \cite{OS11:chow} for a more detailed and extensive account).  Motivated by applications in electrical engineering, the Chow Parameters Problem was intensively studied in the
1960s and early 1970s;
several researchers suggested heuristics of various sorts
\cite{Kaszerman:63,Winder:63,KaplanWinder:65,Dertouzos:65} which
were experimentally analyzed in \cite{Winder:69}. See
\cite{Winder:71} for a survey covering much of this early work and
\cite{Baugh:73, Hurst:73} for some later work from this
period.

Researchers in game theory and voting theory rediscovered Chow's theorem
in the 1970s \cite{Lapidot:72}, and the theorem and related results
have been the subject of study in those communities down to the present
\cite{DubeyShapley:79,EinyLehrer:89,TaylorZwicker:92,Freixas:97,
Leech:03,Carreras:04,FM:04,TT:06,APL:07}.  Since the Fourier
coefficient $\wh{f}(i)$ can be viewed as representing the ``influence''
of the $i$-th voter under voting scheme $f$ (under the ``Impartial Culture
Assumption'' in the theory of social choice, corresponding
to the uniform distribution over inputs $x \in \{-1,1\}^n$), the
Chow Parameters Problem corresponds to designing
a set of weights for $n$ voters so that each individual voter has a
certain desired level of influence over the final outcome.

In the 1990s and 2000s several researchers in learning theory
considered the Chow Parameters Problem.
Birkendorf et al.\ \cite{BDJ+:98} showed that the
Chow Parameters Problem
is equivalent to the problem of efficiently learning
LTFs under the uniform distribution in the ``1-Restricted
Focus of Attention (1-RFA)''
model of Ben-David and Dichterman \cite{BenDavidDichterman:98}
(we give more details on this learning model in
Section~\ref{sec:learn}).
Birkendorf et al.\ showed that if $f$ is an LTF with integer weights
of magnitude at most $\poly(n)$,
then estimates of the Chow parameters that are
accurate to within an additive $\pm \eps/\poly(n)$ information-theoretically
suffice to specify the halfspace $f$ to within $\eps$-accuracy.
Other information-theoretic results of this flavor
were given by \cite{Goldberg:06b,Servedio:07cc}.
In complexity theory several generalizations of
Chow's Theorem were given in \cite{Bruck:90,RSO+:95}, and
the Chow parameters play an important role in a recent study
\cite{CHIS:10} of the approximation-resistance of linear threshold predicates
in the area of hardness of approximation.

Despite this considerable interest in the Chow Parameters Problem from
a range of different communities, the first provably effective
and efficient algorithm for the Chow Parameters Problem was only
obtained fairly recently.
\cite{OS11:chow} gave a
$\poly(n) \cdot 2^{2^{\tilde{O}(1/\eps^2)}}$-time algorithm which, given
sufficiently accurate estimates of the Chow parameters of an unknown
$n$-variable LTF $f$, outputs an LTF $f'$ that has
$\Pr[f(x) \neq f'(x)] \leq \eps.$

\medskip

\subsection{Our results.} \label{ssec:results}

In this paper we give a significantly improved algorithm for the Chow Parameters Problem, whose running time dependence on $\eps$ is almost doubly exponentially better than the \cite{OS11:chow} algorithm.  Our main result is the following:

\begin{theorem} [Main, informal statement] \label{thm:main}
There is an $\tilde{O}(n^2) \cdot (1/\eps)^{O(\log^2(1/\eps))}  \cdot \log(1/\delta)$-time algorithm $\alg$ with the following property:
 Let $f : \bn \to \bits$ be an LTF and let $0 <
\eps,\delta < 1/2$.  If $\alg$ is given as input $\eps,\delta$ and (sufficiently precise estimates of) the Chow parameters of $f$, then $\alg$ outputs integers $v_1,\dots,v_n,\theta$ such that with probability at least $1 - \delta$, the linear threshold function
$f^{\ast}=\sign(v_1 x_1 + \cdots + v_n x_n - \theta)$
satisfies $\Pr_x [f(x) \neq f^*(x)] \leq \eps.$
\end{theorem}

Thus we obtain an efficient randomized polynomial approximation scheme (ERPAS)
with a {\em quasi-polynomial} dependence on $1/\eps$. We note that for the subclass of LTFs
with integer weights of magnitude at most $\poly(n)$,
our algorithm runs in $\poly(n/\eps)$ time, i.e. it is a {\em fully} polynomial randomized approximation scheme (FPRAS)
(see Section~\ref{ssec:main-proofs} for a formal statement).
Even for this restricted subclass of LTFs, the algorithm of~\cite{OS11:chow} runs in time doubly exponential in $1/\eps$.

Our main result has a range of interesting implications in learning theory. First, it directly gives an efficient algorithm for learning LTFs in the uniform distribution $1$-RFA
model. Second, it yields a very fast agnostic-type algorithm for learning LTFs in the standard uniform distribution PAC model. Both these algorithms run in time quasi-polynomial in $1/\eps$. We elaborate on these learning applications in Section~\ref{sec:learn}.

An interesting feature of our algorithm is that it outputs an LTF with integer weights of magnitude at most $\sqrt{n} \cdot (1/\eps)^{O(\log^2(1/\eps))}$. Hence, as a corollary of our approach, we obtain essentially optimal bounds on approximating arbitrary LTFs using LTFs with small integer weights.  It has been known since the 1960s that every $n$-variable LTF $f$ has an exact representation $\sign(w \cdot x - \theta)$ in which all the weights $w_i$ are integers satisfying $|w_i| \leq 2^{O(n \log n)}$, and H{\aa}stad \cite{Hastad:94} has shown that there is an $n$-variable LTF $f$ for which {\em any} integer-weight representation must have each
$|w_i| \geq2^{\Omega(n \log n)}.$  However, by settling for an approximate representation
(i.e. a representation $f' = \sign(w \cdot x - \theta)$ such that $\Pr_{x} [f(x) \neq f'(x)] \leq \eps$), it is possible to get away with much smaller integer weights.  Servedio \cite{Servedio:07cc} showed that every LTF $f$ can be $\eps$-approximated using integer weights each at most $\sqrt{n} \cdot 2^{\tilde{O}(1/\eps^2)}$, and this bound was subsequently improved (as a function of $\eps$) to $n^{3/2} \cdot 2^{\tilde{O}(1/\eps^{2/3})}$ in \cite{DiakonikolasServedio:09}.  (We note that ideas and tools that were developed in work on low-weight approximators for LTFs have proved useful in a range of other contexts, including hardness of approximation \cite{FGRW09}, property testing \cite{MORS:10}, and explicit constructions of pseudorandom objects \cite{DGJ+:10}.)

Formally, our approach to proving Theorem~\ref{thm:main} yields the following nearly-optimal weight bound on $\eps$-approximators for LTFs:

\begin{theorem} [Low-weight approximators for LTFs] \label{thm:lowwt}
Let $f: \bn \to \bits$ be any LTF.  There is an LTF
$f^{\ast} = \sign(v_1 x_1 + \cdots + v_n x_n - \theta)$ such that $\Pr_x [f(x) \neq f^{\ast}(x)] \leq \eps$ and the weights $v_i$ are integers
that satisfy %$\max_i |v_i| \leq \sqrt{n} \cdot (1/\eps)^{O(\log^2(1/\eps))}$ and
$$\littlesum_{i=1}^n v_i^2 = n \cdot  (1/\eps)^{O(\log^2(1/\eps))}.$$
\end{theorem}

The bound on the magnitude of the weights in the above theorem is optimal as a function of $n$ and nearly optimal as a function of $\eps$.
Indeed, as shown in~\cite{Hastad:94,Goldberg:06b}, in general any $\eps$-approximating LTF $f^{\ast}$ for an arbitrary $n$-variable
LTF $f$ may need to have integer weights at least $\max\{\Omega(\sqrt{n}),(1/\eps)^{\Omega(\log\log(1/\eps))}\}$.  Thus,
Theorem~\ref{thm:lowwt} nearly closes what was previously an almost exponential gap  between the known upper and lower bounds for this problem.
Moreover, the proof of Theorem~\ref{thm:lowwt} is constructive (as opposed e.g. to the one in~\cite{DiakonikolasServedio:09}),
i.e. there is a randomized $\poly(n) \cdot  (1/\eps)^{O(\log^2(1/\eps))}$-time algorithm that constructs an $\eps$-approximating LTF.

\medskip

\noindent {\bf Techniques.}  We stress that not only are the quantitative results of
Theorems~\ref{thm:main} and~\ref{thm:lowwt} dramatically stronger than
previous work, but the proofs are significantly more self-contained and
elementary as well. The \cite{OS11:chow}
algorithm relied heavily on several rather sophisticated
results on spectral properties of linear threshold functions; moreover, its proof
of correctness required a careful re-tracing of the
(rather involved) analysis of a fairly complex property testing algorithm
for linear threshold functions given in \cite{MORS:10}.  In contrast,
our proof of Theorem~\ref{thm:main} entirely bypasses these spectral
results and does not rely on \cite{MORS:10} in any way.
Turning to low-weight approximators, the improvement
from $2^{\tilde{O}(1/\eps^2)}$ in \cite{Servedio:07cc} to
$2^{\tilde{O}(1/\eps^{2/3})}$ in \cite{DiakonikolasServedio:09}
required a combination of  rather delicate linear programming arguments and
powerful results on the anti-concentration of sums of independent
random variables due to Hal{\'a}sz \cite{Hal:77}.  In contrast, our proof of Theorem~\ref{thm:lowwt}
bypasses anti-concentration entirely and does not require any sophisticated linear programming arguments.

%\medskip

Two main ingredients underlie the proof of Theorem~\ref{thm:main}.
The first is a new structural result relating the ``Chow distance''
and the ordinary (Hamming) distance between two functions $f$ and $g$, where $f$ is an LTF and $g$ is an arbitrary bounded function.
The second is a new and simple algorithm which, given (approximations to) the Chow parameters of an arbitrary Boolean function $f$,
efficiently construct a ``linear bounded function'' (LBF) $g$ --  a certain type of bounded function --
whose ``Chow distance'' from $f$ is small.
We describe each of these contributions in more detail below.

\subsection{The main structural result.}\label{sec:structuralresult}

In this subsection we first give the necessary definitions regarding Chow parameters and Chow distance,
and then state Theorem~\ref{thm:dchow-vs-dh}, our main structural
result.

\subsubsection{Chow parameters and distance measures.}

We formally define the Chow parameters of a function on $\{-1,1\}^n$:

\begin{definition} \label{def:chow}
Given any function $f : \bn \to \R$, its \emph{Chow Parameters}
are the rational numbers $\wh{f}(0), \wh{f}(1), \dots, \wh{f}(n)$ defined by
$ \wh{f}(0) = \E[f(x)],$ $\wh{f}(i) = \E[f(x)x_i]$  for $1 \leq i \leq n$.
We say that the \emph{Chow vector} of $f$ is $ \mychows_f = (\wh{f}(0), \wh{f}(1), \dots, \wh{f}(n)).$
\end{definition}

The Chow parameters naturally induce a distance measure between functions $f,g$:

\begin{definition}  Let $f, g : \bn \to \R$. We define
the \emph{Chow distance} between $f$ and $g$ to be
$ \dchow(f,g) \eqdef  \| \mychows_f  - \mychows_g \|_2$,
i.e. the Euclidean distance between the Chow vectors.
\end{definition}

This is in contrast with the familiar $L_1$-distance between functions:

\begin{definition}
The \emph{distance}
between  two functions
$f, g : \bn \to \R$ is defined as $\dist(f,g) \eqdef \E [|f(x)-g(x)|].$
If $\dist(f,g) \leq \eps$, we say that $f$ and $g$ are \emph{$\eps$-close}.
\end{definition}

We note that if $f,g$ are Boolean functions with range $\{-1,1\}$ then $\dist(f, g)=2\Pr[f(x) \neq g(x)]$ and thus $\dist$ is equivalent (up to a factor of 2) to the familiar Hamming distance.
\ignore{(In the
 proof of Theorem~\ref{thm:dchow-vs-dh} we will eventually need to generalize $\dchow$ and $\dist$ to
functions over domains $\{-1,1\}^i \times \R^{j}$, but we defer this until later.)}

\subsubsection{The main structural result:  small Chow-distance implies small distance.}
\label{sec:mainresult}

The following fact can be proved easily using basic Fourier analysis (see~Proposition 1.5 in~\cite{OS11:chow}):

\begin{fact} \label{fact:distances}
Let $f, g: \bn \to \R.$ We have that $\dchow(f,g) \leq 2 \sqrt{\dist(f,g)}$.
\end{fact}

Our main structural result, Theorem~\ref{thm:dchow-vs-dh}, is essentially a converse which
bounds $\dist(f,g)$ in terms of $\dchow(f, g)$ when $f$ is an LTF and $g$ is any bounded function:

\begin{theorem} [Main Structural Result] \label{thm:dchow-vs-dh}
Let $f : \bn$ $\rightarrow \bits$ be an LTF and  $g :  \bn \to [-1,1]$ be any bounded function.
If $\dchow(f,g) \le \epsilon$ then $$\dist(f,g) \leq 2^{-\Omega \left( \sqrt[3]{\log(1/\epsilon)} \right)}.$$ \end{theorem}

Since Chow's theorem says that if $f$ is an LTF and $g$ is any bounded
function then $\dchow(f,g)=0$ implies that $\dist(f,g)=0,$
Theorem~\ref{thm:dchow-vs-dh} may be viewed as a ``robust'' version of Chow's Theorem.
Note that the assumption that $g$ is bounded is necessary for the above statement, since the function
$g(x) = \littlesum_{i=0}^n \wh{f}(i) x_i$ (where $x_0 \equiv 1$) has $\dchow(f,g)=0$, but may have $\dist(f,g) = \Omega(1)$.
Results of this sort but with weaker quantitative bounds were given earlier in
\cite{BDJ+:98,Goldberg:06b,Servedio:07cc,OS11:chow}; we discuss the relationship between Theorem~\ref{thm:dchow-vs-dh} and some of this prior work below.

\medskip

\noindent {\bf Discussion.}  Theorem~\ref{thm:dchow-vs-dh} should be contrasted with Theorem~1.6 of \cite{OS11:chow},
the main structural result of that paper. That theorem says that for $f: \bn \to \bits$ any LTF and $g: \bn \to [-1,1]$
any bounded function\footnote{The theorem statement in \cite{OS11:chow} actually
requires that $g$ have range $\{-1,1\}$, but the proof is easily seen to extend to $g: \{-1,1\}^n \to
[-1,1]$ as well.}, if $\dchow(f,g) \leq \eps$ then $\dist(f,g) \leq \tilde{O}(1/\sqrt{\log(1/\eps)})$.
Our new Theorem~\ref{thm:dchow-vs-dh} provides a bound on
$\dist(f,g)$ which is almost exponentially stronger than the \cite{OS11:chow} bound.

Theorem~\ref{thm:dchow-vs-dh} should also be contrasted with Theorem~4
(the main result) of \cite{Goldberg:06b},
which says that for $f$ an $n$-variable LTF and $g$ any Boolean
function, if $\dchow(f,g) \leq (\eps/n)^{O(\log(n/\eps)\log(1/\eps))}$
then $\dist(f,g) \leq \eps.$ Phrased in this way, Theorem~\ref{thm:dchow-vs-dh}
says that for $f$ an LTF and $g$ any bounded function, if
$\dchow(f,g) \leq \eps^{O(\log^2(1/ \eps))}$ then $\dist(f,g) \leq \eps$.
So our main structural result may be viewed as an improvement
of Goldberg's result that removes its dependence on $n$.  Indeed, this is not a coincidence;
Theorem~\ref{thm:dchow-vs-dh} is proved by carefully extending and strengthening Goldberg's
arguments using the ``critical index'' machinery developed in recent studies of structural properties
of LTFs \cite{Servedio:07cc,OS11:chow,DGJ+:10}.

It is natural to wonder whether the conclusion of
Theorem~\ref{thm:dchow-vs-dh} can be strengthened to ``$\dist(f,g)
\leq \eps^c$'' where $c>0$ is some absolute constant.  We show that
no such strengthening is possible, and in fact, no conclusion of the
form ``$\dist(f,g) \leq 2^{-\gamma(\eps)}$'' is possible for any function
$\gamma(\eps) = \omega(\log(1/\eps)/\log\log(1/\eps))$; we prove this in
Section~\ref{sec:nearoptchow}.

\subsection{The algorithmic component.} \label{ssec:alg}

A straightforward inspection of the arguments in \cite{OS11:chow} shows that by
using our new Theorem~\ref{thm:dchow-vs-dh} in place of Theorem~1.6 of that paper
throughout, the running time of the \cite{OS11:chow}
algorithm can be improved to $\poly(n) \cdot 2^{(1/\eps)^{O(\log^2(1/\eps))}}.$
This is already a significant improvement over the
$\poly(n) \cdot 2^{2^{\tilde{O}(1/\eps^2)}}$ running time of
\cite{OS11:chow}, but is significantly worse than the
$\poly(n) \cdot (1/\eps)^{O(\log^2(1/\eps))}$ running time which is
our ultimate goal.

The second key ingredient of our results is a new algorithm for constructing an LTF
from the (approximate) Chow parameters of an LTF $f$. The previous approach to this problem \cite{OS11:chow} constructed an LTF with Chow parameters close to $\vec{\chi}_f$ directly and applied the structural result to the constructed LTF. Instead, our approach is based on the insight that it is substantially easier to find a bounded real-valued function $g$ that is close to $f$ in Chow distance. The structural result can then be applied to $g$ to conclude that $g$ is close to $f$ in $L_1$-distance.
The problem with this idea is, of course, that we need an LTF that is close to $f$ and not a general bounded function. However, we show that it is possible to find $g$ which is a ``linear bounded function'' (LBF), a type of bounded function closely related to LTFs. An LBF can then be easily converted to an LTF with only a small increase in distance from $f$. We now proceed to define the notion of an LBF and state our main algorithmic result formally.
We first need to define the notion of a projection:
\begin{definition} \label{def:projection}
For a real value $a$, we denote its projection to $[-1,1]$ by $P_1(a)$. That is, $P_1(a) = a$ if $|a| \leq 1$ and $P_1(a) = \sign(a)$, otherwise.
\end{definition}
\begin{definition} \label{def:lbf}
A function $g:\{-1,1\}^n \rightarrow [-1,1]$ is referred to as a {\em linear bounded function} (LBF) if there exists a vector of real values $w = (w_0,w_1,\ldots,w_n)$ such that
$g(x)  = P_1(w_0 + \littlesum_{i=1}^n w_i x_i)$. The vector $w$ is said to represent $g$.
\end{definition}

We are now ready to state our main algorithmic result:

\begin{theorem}[Main Algorithmic Result] \label{thm:alg}
There exists a randomized algorithm {\tt ChowReconstruct} that for every Boolean function $f: \{-1,1\}^n \rightarrow \{-1,1\}$,
given $\eps>0,\delta>0$ and a vector $\vec{\alpha} = (\alpha_0,\alpha_1,\ldots,\alpha_n)$ such that $\|\mychows_f - \vec{\alpha}\| \leq \eps$,
with probability at least $1-\delta$, outputs an LBF $g$ such that $\|\mychows_f - \mychows_g\| \leq 6 \eps$.
The algorithm runs in time $\tilde{O}(n^2 \eps^{-4} \log{(1/\delta}))$.
Further, $g$ is represented by a weight vector $ \kappa v \in \R^{n+1}$, where $\kappa \in \R$ and $v$ is an integer vector
of length $\|v\| = O(\sqrt{n}/\eps^3)$.
\end{theorem}

We remark that the condition on the weight vector $v$ given by Theorem \ref{thm:alg} is the key for the proof of Theorem \ref{thm:lowwt}.

Note that the running time of {\tt ChowReconstruct} is polynomial in the relation between Chow distance and $L_1$-distance. By the structural result of~\cite{BDJ+:98}, this implies that for the subclass of LTFs with integer weights of magnitude bounded by  $\poly(n)$, we obtain a $\poly(n/\eps)$ time algorithm, i.e. an FPRAS.

\medskip

\noindent {\bf Discussion.} It is interesting to note that the approach underlying Theorem~\ref{thm:alg} is much more efficient
and significantly simpler than the algorithmic approach of~\cite{OS11:chow}.
The algorithm in~\cite{OS11:chow} roughly works as follows: In the first step, it constructs a ``small''  set of candidate LTFs
such that at least one of them is close to $f$, and in the second step it identifies such an LTF by searching over all such candidates.
The first step proceeds by enumerating over ``all'' possible weights assigned to the ``high influence'' variables.  This brute force search
makes the~\cite{OS11:chow} algorithm very inefficient. Moreover, its proof of correctness requires some sophisticated
spectral results from~\cite{MORS:10}, which make the approach rather complicated.

In this work, our algorithm is based on a boosting-based approach, which is novel in this context.
Our approach is much more efficient than the brute force search of~\cite{OS11:chow} and its analysis is much simpler,
since it completely bypasses the spectral results of~\cite{MORS:10}. We also note that the algorithm of~\cite{OS11:chow} crucially depends
on the fact that the relation between Chow distance and distance has no dependence on $n$. (If this was not the case, the approach would not lead
to a polynomial time algorithm.) Our boosting-based approach is quite robust, as it has no such limitation. This fact is crucial for us to obtain
the aforementioned FPRAS for small-weight LTFs.

While we are not aware of any prior results similar to Theorem~\ref{thm:alg} being stated explicitly,
we note that weaker forms of our theorem can be obtained from known results. In particular, Trevisan \etal \cite{TTV09} describe an algorithm that given oracle access to a Boolean function $f$, $\eps' > 0$, and a set of functions $H = \{h_1, h_2,\ldots h_k\}$, efficiently finds a bounded function $g$ that for every $i \leq n$ satisfies $|\E[f \cdot h_i] - \E[g \cdot h_i]| \leq \eps'$. One can observe that if $H= \{1,x_1,\ldots, x_n\}$, then the function $g$ returned by their algorithm is in fact an LBF and that the oracle access to $f$ can be replaced with approximate values of $\E[f \cdot h_i]$ for every $i$. Hence, the algorithm in~\cite{TTV09}, applied to the set of functions $H = \{1, x_1,x_2,\ldots,x_n\}$, would find an LBF $g$ which is close in Chow distance to $f$. A limitation of this algorithm is that, in order to obtain an LBF which is $\Delta$-close in Chow distance to $f$, it requires that every Chow parameter of $f$ be given to it with accuracy of $O(\Delta/\sqrt{n})$. In contrast, our algorithm only requires that the total distance of the given vector to $\vec{\chi}_f$ is at most $\Delta/6$. In addition, the bound on the integer weight approximation of LTFs that can be obtained from the algorithm in \cite{TTV09} is linear in $n^{3/2}$, whereas we obtain the optimal dependence of $\sqrt{n}$.

The algorithm in \cite{TTV09} is a simple adaptation of the hardcore set construction technique of Impagliazzo \cite{imp95a}. Our algorithm is also based on the ideas from \cite{imp95a} and, in addition, uses ideas from the distribution-specific boosting technique in \cite{Feldman:10ab}.

Our algorithm can be seen as an instance of a more general approach to learning (or approximating) a function that is based on constructing a bounded function with the given Fourier coefficients. Another instance of this new approach is the recent algorithm for learning a certain class of polynomial threshold functions (which includes polynomial-size DNF formulae) from low-degree Fourier coefficients \cite{Feldman:12colt}. We note that the algorithm in \cite{Feldman:12colt} is based on an algorithm similar to ours. However, like the algorithm in \cite{TTV09}, it requires that every low-degree Fourier coefficient be given to it with high accuracy. As a result it would be similarly less efficient in our application.

\smallskip

\noindent {\bf Organization.}
In Section~\ref{sec:prelims} we record some mathematical preliminaries that will be used throughout the paper. In Section~\ref{sec:exact} we present some observations regarding the complexity of solving the Chow parameters problem exactly and  give an LP--based $2^{O(n)}$-time algorithm for it.
Sections~\ref{sec:struct-overview} and~\ref{sec:structural} contain the proof of our main structural result (Theorem~\ref{thm:dchow-vs-dh}).  In Section~\ref{sec:chow-algo} we present our main algorithmic ingredient (Theorem~\ref{thm:alg}).  Section~\ref{sec:payoff} puts the pieces together and proves our main theorem (Theorem~\ref{thm:main}) and our other main result (Theorem~\ref{thm:lowwt}), while Section~\ref{sec:learn} presents the consequences of our results to learning theory. Finally, in Section~\ref{sec:concl} we conclude the paper and present a few interesting research directions.

\section{Mathematical Preliminaries} \label{sec:prelims}

\subsection{Probabilistic Facts.} \label{ssec:prob-basics}

We require some basic probability results including the standard additive Hoeffding bound:
\begin{theorem} \label{thm:chb}
Let $X_1, \ldots, X_n$ be independent random variables such that for
each $j \in [n]$, $X_j$ is supported on $[a_j, b_j]$ for some $a_j,
b_j \in \R$, $a_j \le b_j$. Let $X \ = \littlesum_{j=1}^{n} X_j$.
Then, for any $t>0$,
$\Pr \big[ |X - \E[X]| \ge t \big] \le 2 \exp \left(
-2t^2/\littlesum_{j=1}^{n} (b_j-a_j)^2   \right).$
\end{theorem}

\medskip

\noindent The Berry-Ess{\'e}en theorem (see e.g. \cite{Feller})
gives explicit error bounds for the Central
Limit Theorem:

\begin{theorem} \label{thm:be} (Berry-Ess{\'e}en)
Let $X_1, \dots, X_n$ be independent random variables satisfying
$\E[X_i] = 0$ for all $i \in [n]$, $\sqrt{\littlesum_i \E[X_i^2]} =
\sigma$, and $\littlesum_i \E[|X_i|^3] = \rho_3$.  Let $S = (X_1 +
\cdots + X_n)/\sigma$ and let $F$ denote the cumulative distribution
function (cdf) of $S$. Then
$\sup_x |F(x) - \Phi(x)| \leq \rho_3/\sigma^3$
where $\Phi$ denotes the cdf of the standard gaussian random variable.
\end{theorem}

An easy consequence of the Berry-Ess{\'e}en theorem is the following fact, which says that
a regular linear form has good anti-concentration (i.e. it assigns small probability
mass to any small interval):

\begin{fact} \label{fact:be}
Let $w=(w_1,\dots,w_n)$ be a $\tau$-regular vector in $\R^n$ and write $\sigma$ to denote $\|w\|_2$.
%Let $x_1, \dots, x_n$ denote independent uniformly random $\pm 1$ signs
%and let $w_1, \dots, w_n \in \R$.  Write $\sigma = \sqrt{\littlesum_i
%w_i^2}$, and assume $|w_i| \leq \tau \sigma$ for all $i \in [n]$.
Then for any interval $[a,b] \subseteq \R$, we have
$\big|\Pr[\littlesum_{i=1}^n w_i x_i \in (a,b]] - \Phi([a/\sigma, b/\sigma])\big| \leq 2\tau$,
where $\Phi([c,d]) \eqdef \Phi(d) - \Phi(c)$. In particular, it follows that
$$\Pr \big[\littlesum_{i=1}^n w_i x_i \in (a,b]\big] \leq  |b -a| / \sigma + 2\tau.$$
\end{fact}

\subsection{Useful inequalities.} \label{ssec:basic-ineq}
We will need the following elementary inequalities.
\begin{fact}\label{fac:ineq1}
For $a, b \in (0,1)$, $(ab)^{\log (1/a) + \log (1/b)} \ge a^{2 \log (1/a)} \cdot b^{2\log (1/b)}.$
\end{fact}
\begin{proof}
\begin{eqnarray*} (ab)^{\log (1/a) + \log (1/b)} &=& 2^{- \log^2 (1/a) - \log^2 (1/b) - 2\log (1/a) \cdot \log(1/b) } \\
&\ge& 2^{- 2\log^2 (1/a) - 2\log^2 (1/b) }  \\
&=& a^{2 \log (1/a)} \cdot b^{2\log (1/b)},
\end{eqnarray*} where the inequality is the arithmetic-geometric mean inequality.
\end{proof}

Similarly, we obtain:

\begin{fact}\label{fac:ineq2}
For $x,y \ge 1$, $(x+y)^{-\log(x+y)} \ge (2x)^{-\log (2x)} \cdot (2y)^{-\log (2y)}.$
\end{fact}

\subsection{Useful facts about affine spaces.}
A subset  $V \subseteq \mathbb{R}^n$ is said to be an \emph{affine subspace}
if it
is closed under affine combinations of vectors in $V$. Equivalently, $V$ is
an affine subspace of $\mathbb{R}^n$ if $V =X +b$  where $b \in
\mathbb{R}^n$ and $X$ is a linear subspace of $\mathbb{R}^n$.  The
affine dimension of $V$ is the same as the dimension of the linear subspace
$X$.
A hyperplane in $\mathbb{R}^n$ is  an affine space of dimension $n-1$.
Throughout the paper we use bold capital letters such as $\h$ to denote hyperplanes.

In this paper whenever we refer to a ``subspace'' we mean an affine subspace
unless explicitly otherwise indicated.
The dimension of an affine subspace $V$ is denoted by $\dim(V)$. Similarly,
for a set $S \subseteq \mathbb{R}^n$, we write
$\span(S)$ to denote
the affine span of $S$, i.e.
\[
\span(S) = \{s + \sum_{i=1}^m w_i (x^i - y^i) \mid
s, x^i, y^i \in S, w_i \in \R, m \in \mathbb{N}\}.
\]
The following very useful fact about affine spaces was proved by
Odlyzko\cite{Odlyzko:88}.
\begin{fact}\label{fac:affine}\cite{Odlyzko:88}
Any affine subspace of $\mathbb{R}^n$ of dimension $d$ contains at most $2^d$
elements of $\{-1,1\}^n$.\end{fact}

\section{On the Exact Chow Parameters Problem} \label{sec:exact}

In this section we make some observations regarding the complexity of the exact version of the Chow parameters problem and present
a simple (albeit exponential time) algorithm for it, that beats brute-force search.

\subsection{Proof of Chow's Theorem.} \label{ssec:chow}
For completeness we state and prove Chow's theorem here:

\begin{theorem}  [\cite{Chow:61}] \label{thm:chow}  Let $f : \bn \to \bits$ be an LTF and let $g : \bn \to
[-1,1]$ be a bounded function such that $\wh{g}(j) = \wh{f}(j)$ for all $0 \leq j \leq n$.  Then $g = f$.
\end{theorem}

\begin{proof}
Write
$
f(x) = \sgn(w_0 + w_1 x_1 + \cdots + w_n x_n),
$
where the weights are scaled so that $\sum_{j =0}^n w_j^2 = 1$. We may assume without loss of generality that $|w_0 + w_1 x_1 + \cdots + w_n x_n| \neq 0$ for all $x$.  (If this is not the case, first translate the separating hyperplane by slightly perturbing $w_0$ to make it hold; this can be done without changing $f$'s value on any point of $\{-1,1\}^n.$)  Now we have
\begin{eqnarray*}
0 &=& \sum_{j =0}^n w_j (\widehat{f}(j) - \widehat{g}(j)) \\
&=& \E[(w_0 + w_1 x_1 + \cdots + w_n x_n)(f(x) - g(x))] \\
&=& \E[|f(x) - g(x)| \cdot |w_0 + w_1 x_1 + \cdots + w_n x_n|].
\end{eqnarray*}
The first equality is by the assumption that $\wh{f}(j) = \wh{g}(j)$
for all $0 \leq j \leq n$, the second equality is linearity of expectation (or Plancherel's
identity), and the third equality uses the fact that $$\sign(f(x)-g(x)) = f(x) = \sgn(w_0
+ w_1 x_1 + \cdots + w_n x_n)$$ for any bounded function $g$ with range $[-1,1]$. But since $|w_0 + w_1 x_1 + \cdots + w_n x_n|$ is
always strictly positive,  we must have $\Pr[f(x) \neq g(x)] = 0$ as
claimed.
\end{proof}

\subsection{An exact $2^{O(n)}$--time algorithm.} \label{ssec:LP}

Let us start by pointing out that it is unlikely that the Chow Parameters problem can be solved exactly in polynomial time.
Note that even checking the correctness of a candidate solution is
$\sharp P$-complete, because computing $\wh{f}(0)$ is equivalent to counting 0-1 knapsack solutions.
This suggests (but does not logically imply) that the exact problem is intractable;
characterizing its complexity is an  interesting open problem (see Section~\ref{sec:concl}).

The naive brute-force approach (enumerate all possible $n$-variable LTFs, and for each one check whether
it has the desired Chow parameters) requires $2^{\Theta(n^2)}$ time. The following proposition gives an improved (albeit exponential time) algorithm:

\begin{proposition} \label{prop:exact-chow}
The Chow parameters problem can be solved exactly in time $2^{O(n)}$.
\end{proposition}
\begin{proof}
Let $\alpha_i$, $i=0,1,\ldots, n$ be the target Chow parameters; we are given the promise that there exists an LTF $f:\bn \to \bits$ such that $\wh{f}(i) = \alpha_i$ for all $i$.
Our goal is to output (a weights-based representation of) the function $f$. Let $g:\bn \to [-1,1]$ be a bounded function that has the same Chow parameters as $f$.
We claim that there exists a linear program with $2^n$ variables and $O(2^n)$ constraints encoding the truth-table of $g$.
Indeed, for every $x \in \bn$ we have a variable $g(x)$ and the constraints are as follows: For all $x \in \bn$ we include the constraint $-1 \leq g(x) \leq 1$. We also include
the $(n+1)$ constraints $\E_x [g(x) x_i] \equiv 2^{-n} \sum_{x \in \bn} g(x)x_i = \alpha_i$, $i=0,1,\ldots, n$ (where $x_0 \equiv 1$). Chow's theorem stated above implies that the aforementioned linear program has a {\em unique} feasible solution, corresponding to the truth table of the target LTF $f$.
That is, the unique solution of the linear program will be integral and is identical to the target function.
Since the size of the linear program is $2^{O(n)}$ and linear programming is in P, the truth table of $f$ can thus be computed in time $2^{O(n)}$.

A weight-based representation of $f$ as $\sign(w\cdot x - \theta)$ can then be obtained
straightforwardly in time $2^{O(n)}$ by solving another linear program with variables $(w, \theta)$ and $2^n$ constraints, one for each $x \in \bn$.
\end{proof}

\section{Proof overview of main structural result:  Theorem~7} \label{sec:struct-overview}

In this section we provide a detailed overview of the proof of Theorem~\ref{thm:dchow-vs-dh}, restated here for convenience:

\medskip

\noindent {\bf Theorem~\ref{thm:dchow-vs-dh}} (Main Structural Result){\bf.} \emph{
Let $f : \bn \rightarrow \bits$ be an LTF and  $g :  \bn \to [-1,1]$ be any bounded function.
If $\dchow(f,g) \le \epsilon$ then $\dist(f,g) \leq 2^{-\Omega \left( \sqrt[3]{\log(1/\epsilon)} \right)}$.}

\medskip

We give an informal overview of the main ideas of the proof of Theorem~\ref{thm:dchow-vs-dh} in Section~\ref{sec:outline}, and then proceed with a detailed outline of Theorem~\ref{thm:dchow-vs-dh} in Section~\ref{sec:dchow-vs-dhproof}.

\subsection{Informal overview of the proof.} \label{sec:outline}

We first note that throughout the informal explanation given in this
subsection, for the sake of clarity we restrict our attention
to the case in which $g: \{-1,1\}^n \to \{-1,1\}$
is a Boolean rather than a bounded function.
In the actual proof we deal with bounded functions using a suitable
weighting scheme for points of $\{-1,1\}^n$ (see the
discussion before Fact~\ref{fac:large} near the start
of the proof of Theorem~\ref{thm:dchow-vs-dh}).

To better explain our approach,  we begin with a few words about how Theorem~1.6 of \cite{OS11:chow} (the only previously known statement of this type that is ``independent of $n$'')
is proved.  The key to that theorem is a result on approximating LTFs using LTFs with ``good anti-concentration''; more precisely, \cite{OS11:chow} shows  that
for any LTF $f$ there is an LTF $f'(x)=\sign(v \cdot x - \nu), \|v\|=1,$ that is extremely close to $f$ (Hamming distance roughly $2^{-1/\eps}$) and which has ``moderately good anticoncentration at radius $\eps$,'' in the sense
that $\Pr[|v \cdot x - \nu| \leq \eps] \leq \tilde{O}(1/\sqrt{\log(1/\eps)}).$  Given this, Theorem~1.6 of
\cite{OS11:chow} is proved using a modification of the proof of the original Chow's Theorem.  However, for this approach  based on the original Chow proof to work, it is crucial that the Hamming distance between $f$ and $f'$ (namely $2^{-1/\eps}$) be very small compared to the anti-concentration radius (which is $\eps$).  Subject to this constraint it seems very difficult to give a significant quantitative improvement of the approximation result in a way that would improve the bound of Theorem~1.6 of \cite{OS11:chow}.

Instead, we hew more closely to the approach used to prove Theorem~4 of \cite{Goldberg:06b}.  This approach also involves a perturbation of the LTF $f$, but instead of measuring closeness in terms of Hamming distance, a more direct geometric view is taken.  In the rest of this subsection we give a high-level explanation of Goldberg's proof and of how we modify it to obtain our improved bound.

The key to Goldberg's approach is a (perhaps surprising) statement about the geometry of hyperplanes
as they relate to the Boolean hypercube.  He establishes the following key geometric result (see Theorem~\ref{thm:goldberg-thm3} for a precise statement):

\begin{quote}
If $\h$ is any $n$-dimensional hyperplane such that an $\alpha$ fraction of points in $\{-1,1\}^n$ lie ``very close'' in Euclidean distance (essentially $1/\quasipoly(n/\alpha)$) to $\h$, then there is a hyperplane $\h'$ which actually \emph{contains} all those $\alpha 2^n$ points of the hypercube.
\end{quote}
With this geometric statement in hand, an iterative argument is used to show that if the Hamming distance between LTF $f$ and Boolean function $g$ is large, then
the Euclidean distance between the centers of mass of (the positive examples for $f$ on which $f$ and $g$ differ) and (the
negative examples for $f$ on which $f$ and $g$ differ) must be large; finally, this Euclidean distance between centers of mass corresponds closely to the Chow distance between $f$ and $g$.

\ignore{Intuitively, Goldberg's approach requires less ``control'' over $\h'$ than the \cite{OS11:chow} approach because it does not make any requirement that points be classified correctly by
the LTF corresponding to $\h'$.\anote{I do not find this statement to be accurate. We can probably talk about it.}}
However, the $1/\quasipoly(n)$ closeness requirement in the
key geometric statement means that Goldberg's Theorem~4 not only depends on $n$, but this dependence
is superpolynomial.  The heart of our improvement is to combine Goldberg's key geometric statement with
ideas based on the ``critical index''  of LTFs to get a version of
the statement which is completely independent of $n$.  Roughly speaking, our analogue of Goldberg's key geometric statement is the following (a precise version is given as Lemma~\ref{lem:1} below):

\begin{quote}
If $\h$ is any $n$-dimensional hyperplane such that an $\alpha$ fraction of points in $\{-1,1\}^n$ lie within Euclidean distance $\alpha^{O(\log(1/\alpha))}$ of  $\h$, then there is a hyperplane $\h'$ which contains \emph{all but a tiny fraction} of those $\alpha 2^n$ points of the hypercube.
\end{quote}

Our statement is much stronger than Goldberg's in that there is no dependence on $n$ in the
distance bound from $\h$, but weaker in that we do not guarantee $\h'$ passes through every point; it may miss a tiny fraction of points, but we are able to handle this in the subsequent analysis.
Armed with this improvement, a careful sharpening of Goldberg's iterative argument (to get rid of another
dependence on $n$, unrelated to the tiny fraction of points missed by $\h'$) lets us prove Theorem~\ref{thm:dchow-vs-dh}.

\subsection{Detailed outline of the proof.} \label{sec:dchow-vs-dhproof}

% Title was: Chow vs. Hamming distance for bounded functions: Proof of Theorem~\ref{thm:dchow-vs-dh}

As discussed in Section~\ref{sec:outline}, the key to proving Theorem~\ref{thm:dchow-vs-dh} is
an improvement of Theorem 3 in~\cite{Goldberg:06b}.

\begin{definition}
Given a hyperplane $\h$ in $\R^n$ and $\beta>0$, the \emph{$\beta$-neighborhood} of $\h$ is defined as the set of points in $\R^n$ at Euclidean distance
at most $\beta$ from $\h$.
\end{definition}

We recall the following fact which shows how to express the Euclidean distance of a point from a hyperplane using the standard representation of the hyperplane:
\begin{fact}\label{fac:hyper-dist}
Let  $\h = \{x : w \cdot x - \theta =0\}$ be a hyperplane in $\mathbb{R}^n$ where $\Vert w \Vert =1$. Then for any $x \in \mathbb{R}^n$, the Euclidean distance $d(x,\h)$ of $x$ from $\h$ is $|w \cdot x -\theta|$.
\end{fact}

\begin{theorem}[Theorem~3 in~\cite{Goldberg:06b}] \label{thm:goldberg-thm3}
Given any hyperplane in $\R^n$ whose $\beta$-neighborhood contains a subset $S$ of vertices of $\{-1,1\}^n$,
where $|S| = \alpha \cdot 2^n$, there exists a hyperplane which contains all elements of $S$ provided that
$$0 \leq \beta \leq \left( (2/\alpha) \cdot n^{5+ \lfloor \log(n/\alpha) \rfloor} \cdot (2+\lfloor \log(n/\alpha) \rfloor)! \right)^{-1}.$$
\end{theorem}

 Before stating our improved version of the above theorem, we define the set $U = \cup_{i =1}^n \mathbf{e}_i \cup \mathbf{0}$ where $\mathbf{0} \in \R^n$ is the all zeros vector and $\mathbf{e}_i \in \R^n$ is the unit vector in the
$i^{th}$ direction.

Our improved version of Theorem~\ref{thm:goldberg-thm3} is the following:
\begin{lemma} \label{lem:1}
Let $\h$ be a hyperplane in $\R^n$ whose $\beta$-neighborhood contains a subset $S$ of vertices of $\{-1,1\}^n$,
where $|S| = \alpha \cdot 2^n$. Fix $0 < \kappa < \alpha/2$. Then there exists a hyperplane $\h'$ in $\R^n$
that contains a subset $S^{\ast} \subseteq S$ of cardinality at least $(\alpha-\kappa)\cdot 2^n$ provided that
$ 0 \leq \beta \leq \beta_0$, where
$$\beta_0 \eqdef   (\log(1/\kappa))^{-1/2 } \cdot (\log \log (1/\kappa))^{-O(\log \log \log (1/\kappa))} \cdot  \alpha^{O(\log(1/\alpha))}.$$
Moreover, the coefficient vector defining $\h'$ has at most $$O \left( (1/\alpha^2) \cdot \left( \log \log(1/\kappa) + \log^2(1/\alpha) \right) \right) $$
nonzero coordinates.
Further, for any $x \in U$, if $x$ lies on $\h$ then $x$ lies on $\h'$ as well.
\ignore{ and each coefficient is an integer multiple of  $1/K$ where $K \eqdef (\alpha-\kappa)^{O \left( \log(1/(\alpha-\kappa)) \right)}.$
}
\end{lemma}

%whose $\beta$-neighborhood contains a subset $S$ of vertices of $\{-1,1\}^n$,
%where $|S| = \alpha \cdot 2^n$. Fix $0 < \kappa < \alpha/2$. Then there exists a hyperplane $\h'$ in $\R^n$
%that contains a subset $S^{\ast} \subseteq S$ of cardinality at least $(\alpha-\kappa)\cdot 2^n$ provided that
%$$ 0 \leq \beta \leq \beta_0 \eqdef   (\log(1/\kappa))^{-1/2 } \cdot (\log \log (1/\kappa))^{-O(\log \log \log (1/\kappa))} \cdot  \alpha^{O(\log(1/\alpha))}.$$
%Moreover, the coefficient vector defining $\h'$ has at most $O \left( (1/\alpha^2) \cdot \left( \log \log(1/\kappa) + \log^2(1/\alpha) \right) \right) $
%nonzero coordinates.\ignore{ and each coefficient is an integer multiple of  $1/K$ where $K \eqdef (\alpha-\kappa)^{O \left( \log(1/(\alpha-\kappa)) \right)}.$
%}.
%\end{theorem}

\medskip

\noindent {\bf Discussion.}  We note that while Lemma~\ref{lem:1} may appear to be incomparable to Theorem~\ref{thm:goldberg-thm3}
because it ``loses'' $\kappa 2^n$ points from the set $S$, in fact by taking $\kappa = 1/2^{n+1}$ it must be
the case that our $S^*$ is the same as $S$, and with this choice of $\kappa$, Lemma~\ref{lem:1} gives a strict quantitative improvement of Theorem~\ref{thm:goldberg-thm3}.
(We stress that for our application, though, it will be crucial for us to
use Lemma~\ref{lem:1} by setting the $\kappa$ parameter to depend only on $\alpha$ independent of $n$.)  We further
note that in any statement like Lemma~\ref{lem:1} that does not ``lose'' any points from $S$, the bound
on $\beta$ must necessarily depend on $n$; we show this in Appendix~\ref{ap:optimalbeta}.
Finally, the condition at the end of Lemma~\ref{lem:1} (that if $x \in U$ lies on $\h$, then it lies on $\h'$ as well) is something we will require later for technical reasons.

\smallskip

We give the detailed proof of Lemma~\ref{lem:1} in Section~\ref{ap:lem1}.
We now briefly sketch the main idea underlying the proof of the lemma.
At a high level, the proof  proceeds by  reducing  the
number of variables from $n$ down to
$$m \eqdef O \left( (1/\alpha^2) \cdot (\log(1/\beta) + \log \log (1/\kappa))
\right)$$
followed by an application of Theorem~\ref{thm:newgoldberg3}, a technical
generalization of Theorem~\ref{thm:goldberg-thm3} proved in Appendix~\ref{sec:usefulgoldberg},  in
$\R^m.$ (As we will see later, we use Theorem~\ref{thm:newgoldberg3} instead of Theorem~\ref{thm:goldberg-thm3}
because we need to ensure that points of $U$ which lie on $\h$ continue to lie
on $\h'$.)  The reduction uses the notion of the $\tau$-critical index
applied to the vector $w$ defining $\h.$ (See Section~\ref{ssec:tools} for the relevant definitions.)

The idea of the proof is that for coordinates $i$ in the ``tail'' of $w$
(intuitively, where $|w_i|$ is small) the value of $x_i$
does not have much effect on $d(x,\h)$, and consequently
the condition of the lemma must hold true in a space of much lower dimension than $n$.
To show that tail coordinates of $x$ do not have much effect on $d(x,\h)$, we do a case analysis
based on the $\tau$-critical index $c(w,\tau)$ of $w$ to show that (in both cases)
the $2$-norm of the entire ``tail'' of $w$ must be small.  If $c(w,\tau)$ is large, then this fact follows easily
by properties of the $\tau$-critical index. On the other hand, if $c(w,\tau)$ is small we argue by contradiction as follows:
By the definition of the $\tau$-critical index and the Berry-Ess{\'e}en theorem,
the ``tail'' of $w$ (approximately) behaves like a normal random variable with standard deviation equal to its $2$-norm.
Hence, if the $2$-norm was large, the entire linear form $w \cdot x$ would have good anti-concentration,
which would contradict the assumption of the lemma.
Thus in both cases, we can essentially
ignore the tail and make the effective number of variables be $m$ which is independent of $n$.

\medskip

As described earlier, we view the geometric Lemma~\ref{lem:1} as the key to the proof
of Theorem~\ref{thm:dchow-vs-dh}; however, to obtain Theorem~\ref{thm:dchow-vs-dh}
from Lemma~\ref{lem:1} requires a delicate iterative argument, which we
give in full in the following section.  This argument is essentially a refined
version of Theorem~4 of \cite{Goldberg:06b} with two
main modifications:  one is that we generalize the argument to
allow $g$ to be a bounded function rather than a Boolean
function, and the other is that we get rid of various factors of
$\sqrt{n}$ which arise in the \cite{Goldberg:06b} argument (and which would
be prohibitively ``expensive'' for us). We give the detailed proof in Section~\ref{ap:dchow-vs-dh}.

\section{Proof of Theorem~\ref{thm:dchow-vs-dh}} \label{sec:structural}

In this section we provide a detailed proof of our main structural result (Theorem~\ref{thm:dchow-vs-dh}).

\subsection{Useful Technical Tools.} \label{ssec:tools}

As described above, a key ingredient in the proof of Theorem~\ref{thm:dchow-vs-dh} is the notion of the ``critical index'' of an LTF $f$.  The critical index was implicitly introduced and used in
\cite{Servedio:07cc} and was explicitly used in \cite{DiakonikolasServedio:09,DGJ+:10,OS11:chow} and other works. To define the critical index we need to first define ``regularity'':

\begin{definition}[regularity]
Fix $\tau > 0.$  We say that a vector $w = (w_1, \ldots, w_n) \in \R^n$ is \emph{$\tau$-regular} if $\max_{i \in [n]} |w_i| \leq \tau \|w\| = \tau \sqrt{w_1^2 + \cdots + w_n^2}.$
A linear form $w \cdot x$ is said to be $\tau$-regular if $w$ is $\tau$-regular, and similarly
an LTF is said to be $\tau$-regular
if it is of the form $\sgn(w \cdot x  -\theta)$  where $w$ is $\tau$-regular.
\end{definition}

Regularity is a helpful notion because if $w$ is $\tau$-regular then the Berry-Ess{\'e}en theorem (stated below) tells us that for uniform $x \in \{-1,1\}^n$, the linear form $w \cdot x$ is ``distributed like a Gaussian up to error $\tau$.''  This can be useful for many reasons; in particular, it will let us exploit the strong anti-concentration properties of the Gaussian distribution.

Intuitively, the critical index of $w$ is the first index $i$ such that
from that point on, the vector $(w_i,w_{i+1},\dots,w_n)$ is regular.  A precise definition follows:

\begin{definition}[critical index]
Given a vector $w \in \R^n$ such that $|w_1| \geq \cdots \geq |w_n| > 0$,
for $k \in [n]$ we denote by $\sigma_k$ the quantity $\sqrt{\littlesum_{i=k}^n w_i^2}$.
We define the \emph{$\tau$-critical index $c(w, \tau)$ of $w$} as the smallest
index $i \in [n]$ for which $|w_i| \leq \tau \cdot \sigma_i$. If
this inequality does not hold for any $i \in [n]$, we define $c(w, \tau) = \infty$.
\end{definition}

The following simple fact states that the ``tail weight'' of the vector $w$  decreases exponentially prior to the critical index:
\begin{fact}\label{fact:small-tail}
For any vector $w = (w_1, \ldots, w_n)$ such that $|w_1| \geq \cdots \geq |w_n| > 0$ and $1 \le a \le c(w,\tau)$, we have $\sigma_a < (1-\tau^2)^{(a-1)/2} \cdot \sigma_1$.
\end{fact}
\begin{proof}
If $a <c(w,\tau)$, then by definition $|w_a| > \tau \cdot \sigma_a$. This implies that $\sigma_{a+1} < \sqrt{1-\tau^2} \cdot \sigma_a$. Applying this inequality repeatedly, we get  that $\sigma_a < (1-\tau^2)^{(a-1)/2} \cdot \sigma_1$ for any  $1 \le a \le c(w,\tau)$.
\end{proof}

%We will also use some standard probabilistic facts and inequalities to prove Theorem~\ref{thm:dchow-vs-dh}
%which are collected in Appendix~\ref{ap:dchow-vs-dhtools}.

\subsection{Proof of Lemma~\ref{lem:1}.} \label{ap:lem1}

Let $0< \tau < \alpha$. Let $\h = \{ x \in \R^n \mid w \cdot x = \theta \}$
where we can assume (by rescaling) that $\|w\|_2=1$
and (by reordering the coordinates) that $|w_1| \geq |w_2| \geq \ldots \geq
|w_n|$. Note that the Euclidean distance of any point $x \in \mathbb{R}^n$ from $\h$ is $|w \cdot x - \theta|$. Let us also define $V \eqdef \h \cap U$.
Set $\tau \eqdef \alpha/4$ (for conceptual clarity we will continue
to use ``$\tau$'' for as long as possible in the arguments below).
We consider the $\tau$-critical index $c(w, \tau)$ of the vector $w \in
\mathbb{R}^n$ and proceed by case analysis based on its value.
Fix the parameter $K_0 \eqdef \Theta \left( (1/\tau^2) \cdot( \log \log (1/\kappa) + \log(1/\beta) ) \right).$

\smallskip

\noindent {\bf Case I:} $c(w, \tau) > K_0$. In this case, we partition $[n]$
into a set of ``head'' coordinates $H = [K_0]$
and a complementary set of ``tail'' coordinates $T = [n] \setminus H$.
Writing $w$ as $(w_H,w_T)$ and likewise for $x$,
it follows from Fact~\ref{fact:small-tail} that
$\| w_{T} \| \leq O(\beta / \sqrt{\log(1/\kappa)})$.  By the Hoeffding bound, for $(1-\kappa)$ fraction of $x \in \bn$ we have that $|w_T \cdot x_T| \leq \beta$.
Therefore, for $(1-\kappa)$ fraction of $x \in \bn$ we have
$$\left| w_{H} \cdot x_{H} - \theta \right| \le \left| w \cdot x -\theta \right| +  \left| w_{T} \cdot x_{T} \right|     \leq  \left| w \cdot x -\theta \right| +  \beta.$$

\medskip
By the assumption of the lemma, there exists a set $S \subseteq \bn$ of cardinality at least $\alpha\cdot 2^n$ such that for all $x \in S$ we have
$|w\cdot x  - \theta| \leq \beta.$
A union bound and the above inequality imply that there exists a set $S^{\ast} \subseteq S$ of cardinality at least $(\alpha - \kappa) \cdot 2^n$ with the property that
for all $x \in S^{\ast}$, we have
$$|w_{H} \cdot x_{H} - \theta|  \le 2\beta.$$

Also, any $x \in U$ satisfies $\Vert x_T \Vert \le 1$. Hence for any $x \in V$, we have that
\begin{eqnarray*}
|w_H \cdot x_H -\theta|  &\leq& |w \cdot x - \theta| + |w_T \cdot x_T| = |w_T \cdot x_T| \\
&\leq& \|w_T\| \cdot \|x_T\| \leq O(\beta / \sqrt{\log(1/\kappa)}) \le \beta.
\end{eqnarray*}

Define the projection mapping $\phi_{H} :  \mathbb{R}^n \to \mathbb{R}^{|H|}$ by
$\phi_{H} : x \mapsto x_{H}$ and consider the image of $S^{\ast}$, i.e.
$S' \eqdef \phi_H (S^{\ast})$. It is clear that
$|S'| \geq (\alpha-\kappa) \cdot 2^{|H|}$ and that for all $x_{H} \in S'$, we have
$$ |w_{H} \cdot x_{H} - \theta|  \leq 2 \beta.$$
Similarly, if $V'$ is the image of $V$ under $\phi_H$, then for every $x_H \in V'$
we have $|w_H \cdot x_H -\theta| \le \beta$.
\smallskip
It is also clear that $\Vert w_T \Vert <1/2$ and hence $\Vert w_H \Vert >1/2$.
Thus for every $x_H \in (S' \cup V')$ we have
$$ \left|\frac{w_{H} \cdot x_{H}}{\Vert w_H \Vert} - \frac{\theta}{\Vert w_H \Vert}\right|  \leq 4 \beta.$$

We now define the $K_0$-dimensional hyperplane $\h_{H}$ as $\h_{H} \eqdef
\{ x_H \in \R^{|H|} \mid w_H \cdot x_H = \theta \}$. As all points in $S' \cup V'$
are in the $4 \beta$-neighborhood of $\h_{H}$, we may now apply
Theorem~\ref{thm:newgoldberg3} for the hyperplane $\h_{H}$ over $\mathbb{R}^{|H|}$
to deduce the existence of an alternate hyperplane $\h'_{H} = \{ x_H \in
\R^{|H|} \mid v_H \cdot x_H = \nu \}$ that contains all points in $S' \cup V'$.
The only condition we need to verify in order
that Theorem~\ref{thm:newgoldberg3} may be applied is that $4\beta$ is upper bounded by
$$ \left( \frac{2}{\alpha-\kappa} \cdot K_0^{5+ \lfloor \log(K_0/(\alpha-\kappa)) \rfloor}  \cdot \left( 2+\lfloor \log\left( K_0 / (\alpha-\kappa)\right) \rfloor  \right)! \right)^{-1}.$$
In the following $C_1,C_2,$ etc. denote unspecified absolute positive
constants.
Using $\kappa < \alpha/2$, it suffices to ensure
$$
\beta < (\alpha/K_0)^{C_1( \log (K_0/\alpha))}.$$
Recalling that $\tau =\alpha/4$ and  plugging in the value of $K_0$ in terms of $\alpha$, $\kappa$ and $\beta$, we need to verify that
$$
\beta < \left( \frac{\alpha^3}{ \log \log (1/\kappa) + \log(1/\beta)} \right)^{
C_2( \log(1/\alpha^3) + \log (\log \log (1/\kappa) + \log(1/\beta)) )}.
$$
Using Fact~\ref{fac:ineq1}, we get that the right hand side is lower bounded
by
$$
\alpha^{C_3 \log(1/\alpha)}  \cdot (\log \log (1/\kappa) + \log(1/\beta))^{
-C_3  \log (\log \log (1/\kappa) + \log(1/\beta))}.
$$
Using Fact~\ref{fac:ineq2}, we get that the above expression is lower bounded
by
$$
\alpha^{C_4 \log(1/\alpha)} \cdot \log \log (1/\kappa)^{-C_4 \log \log \log (1/\kappa) } \cdot  \log(1/\beta)^{ - C_4 \log \log (1/\beta)}.
$$
Thus it suffices to verify that
$$
\beta \le \alpha^{C_4 \log(1/\alpha)} \cdot \log \log (1/\kappa)^{-C_4 \log \log \log (1/\kappa) } \cdot  \log(1/\beta)^{ - C_4 \log \log (1/\beta)}.
$$
It is easy to see that for  $$\beta \le  \alpha^{O(\log(1/\alpha))} \cdot \log \log (1/\kappa)^{-O( \log \log \log (1/\kappa)) }$$  (with sufficiently large constants inside the $O(\cdot)$ notation),  the above inequality is indeed true and hence it is true for $\beta \le \beta_0$.

Thus, we get a new
hyperplane $\h'_{K_0} = \{ x_H \in \R^{|H|} \mid v_H \cdot x_H = \nu \}$
that contains all points in $S' \cup V'$. It is then clear that the $n$-dimensional
hyperplane $\h' = \{ x \in \R^{n} \mid v_H \cdot x_H = \nu \}$ contains
all the points in $S^{\ast} = (\phi_H)^{-1}(S')$ and the points in $V$, and that the vector
$v_H$ defining $\h'$ has the claimed number of nonzero coordinates.
So the theorem is proved in Case I.

%$\h' = \{ x \in \R^{n} \mid v_H \cdot x_H = \nu \}$ contains all the points in $S^{\ast} = (\phi_H)^{-1}(S')$
%We may now apply Theorem~\ref{thm:goldberg-thm3} for the $K_0$-dimensional hyperplane
%$\h_{K_0} = \{ x_H \in \R^{|H|} \mid w_H \cdot x_H = \theta \}$ and deduce that there exists an alternate hyperplane

%$\h'_{K_0} = \{ x_H \in \R^{|H|} \mid v_H \cdot x_H = \nu \}$ that contains all points in $S'$. It is then clear that the $n$-dimensional hyperplane
%$\h' = \{ x \in \R^{n} \mid v_H \cdot x_H = \nu \}$ contains all the points in $S^{\ast} = (\phi_H)^{-1}(S')$. To be able to apply Theorem~\ref{thm:goldberg-thm3}
%we need to ensure that its condition on the upper bound for $\beta$ is satisfied. That is, we need to ensure that
%$$\beta \leq ((\alpha-\kappa)/m)^{O\left( \log \frac{m}{\alpha - \kappa} \right)}.$$
%It is easy to see that this condition is satisfied for $\beta \leq \beta_0$ which completes the analysis of Case I.
%This means that should $\beta \le (\alpha - \kappa) ^{O(\log(1/(\alpha - \kappa)))} \cdot (\log \log (1/\kappa))^{ ( O( \log \log \log (1/\kappa))}$, then we will be fine.

\smallskip

\noindent {\bf Case II:} $c(w, \tau) \leq K_0$. In this case, we partition $[n]$ into ``head''  and ``tail'' based on the value of $c(w, \tau)$  by taking
$H = [c(w, \tau)]$ and $T = [n] \setminus H$. We use the fact that $w_T$
is $\tau$-regular to deduce that the norm of the tail must be small.

\begin{claim}\label{clm:2}
We have $\| w_T \|_2 \le 2\beta/(\alpha - 3\tau) = 8 \beta/\alpha.$
\end{claim}
\begin{proof}
Suppose for the sake of contradiction that $$\|w_T\|_2 > 2\beta/(\alpha -
3\tau).$$
By the Berry-Ess{\'e}en theorem (Theorem~\ref{thm:be}, or more precisely
Fact~\ref{fact:be}), for all $\delta>0$ we have
$${\sup}_{t \in \R} \Pr_{x_T} \left[ \left| w_T \cdot x_T - t \right| \le \delta \right] \leq \frac{2\delta }{\Vert w_T \Vert} + 2\tau.$$
By setting $\delta \eqdef (\alpha-3\tau)\|w_T\| / 2 > \beta$ we get that
$${\sup}_{t \in \R} \Pr_{x_T} \left[ \left| w_T \cdot x_T - t \right| \le
\delta \right] < \alpha,$$
and consequently
\begin{eqnarray*}
\Pr_{x}[ | w \cdot x  - \theta | \le \beta ]
&\leq& \sup_{t \in \R}\Pr_{x_T}[|w_T \cdot x_T - t| \leq \beta] \\
&\leq& \sup_{t \in \R} \Pr_{x_T} \left[ \left| w_T \cdot x_T - t \right| \le \delta \right] \\
&<& \alpha
\end{eqnarray*}
which contradicts the existence of the set $S$ in the statement of the lemma.
\end{proof}

The rest of the proof proceeds similarly to Case I.
By the Hoeffding bound, for $1-\kappa$ fraction of $x \in \bn$ we have
$$|w_{H} \cdot x_{H} - \theta| \le  |w \cdot x -\theta| +  \beta'$$
where $\beta'  = O\left( (\beta/\alpha) \cdot \sqrt{\log(1/\kappa)} \right).$
By the assumption of the lemma and a union bound,
there exists a set $S^{\ast} \subseteq S$ of cardinality at least $(\alpha - \kappa) \cdot 2^n$ with the property that
for all $x \in S^{\ast}$ we have
$$|w_{H} \cdot x_{H} - \theta|  \le \beta' + \beta.$$
Turning to $V$, for every point $x \in V$ we have that $|w_{H} \cdot x_{H} - \theta| \le
|w \cdot x - \theta| + |w_T \cdot x_T| = |w_T \cdot x_T|.$  For $x \in V$ the value
$w_T \cdot x_T$ is either 0 (if $x = \mathbf{0}$) or is $(w_T)_i$ (if $x = \mathbf{e}_i$) for some $i \in T.$ Since $w_T$ is $\tau$-regular we have $|(w_T)_i| \leq \tau \cdot \|w_T\| \leq (\alpha/4) \cdot (8 \beta/\alpha) = 2 \beta$, so for every $x \in V$ we have $|w_H \cdot x_H - \theta| \leq 2 \beta
\leq \beta + \beta'.$

As before, we define the projection mapping $\phi_{H} : \mathbb{R}^n \to\mathbb{R}^{|H|}$ by
$\phi_{H} : x \mapsto x_{H}$. We let  $S' \eqdef
\phi_H (S^{\ast})$ and $V' \eqdef \phi_H(V)$.  It is clear that
$|S'| \geq (\alpha-\kappa) \cdot 2^{|H|}$ and that for all $x_{H} \in (S' \cup V')$
we have
$$ |w_{H} \cdot x_{H} - \theta|  \leq \beta' + \beta.$$
and that for all $x_H \in V'$, $ |w_{H} \cdot x_{H} - \theta|  \leq \beta.$
\smallskip
 We now define the $|H|$-dimensional hyperplane $\h_{H}$ as
$\h_{H} \eqdef \{ x_H \in \R^{|H|} \mid w_H \cdot x_H = \theta \}$.
As before, we note that  $\Vert w_T \Vert <1/2$ and hence
$\Vert w_H \Vert >1/2$. Hence, every point $x_H \in S' \cup V'$ is
$2(\beta + \beta') \le 4 \beta'$ close to $\h_{H}$. As all points in $S' \cup V'$
are $4 \beta'$ close to $\h_{H}$, we may now apply
Theorem~\ref{thm:newgoldberg3} over $\R^{|H|}$ to
deduce the existence of an alternate
hyperplane $\h'_{H} \eqdef \{ x_H \in \R^{|H|} \mid v_H \cdot x_H = \nu \}$
that contains all points in $S'$ and $V'$. The only condition we need to verify is that $4\beta'$ is at most
$$\left( \frac{2}{\alpha-\kappa} \cdot |H|^{5+ \lfloor \log(|H|/(\alpha-\kappa)) \rfloor} \cdot (2+\lfloor \log(|H|/(\alpha-\kappa)) \rfloor)! \right)^{-1}.$$
As $\beta' = O((\beta \sqrt{\log(1/\kappa)})/\alpha)$, doing a calculation akin to the calculation in Case I
(now using $|H| \leq K_0$) we get that the above inequality is true
for $$\beta \le (\log(1/\kappa))^{-1/2} \cdot \alpha^{O(\log(1/\alpha))} \cdot \log \log (1/\kappa)^{-O( \log \log \log (1/\kappa)) }$$ as long as the constant inside the $O(\cdot)$ notation are sufficiently large.
(It is instructive to note here that it is Case II which is the ``bottleneck'' for our overall bound, in the sense that we require a stronger upper bound on $\beta$ for Case II than for Case I.)
%
%\smallskip
%
% Hence, the inequality is true for $\beta\le \beta_0$
%The above condition can be verified following the same proof as in {\bf Case I} and using $|H| \le K_0$.
It is now clear that the $n$-dimensional hyperplane $\h' = \{ x \in \R^{n} \mid v_H \cdot x_H = \nu \}$
contains all the points in $S^{\ast} = (\phi_H)^{-1}(S')$  and the points in $V$, and has the
claimed number of nonzero coordinates.  This proves the Lemma in Case II
and concludes the proof of Lemma~\ref{lem:1}.
%Again, redoing the proof, we will get the condition $\beta \le (\alpha - \kappa) ^{O(\log(1/(\alpha - \kappa)))} /(\log  (1/\kappa))$
%\end{proof}

\subsection{Proof of Theorem~\ref{thm:dchow-vs-dh}.} \label{ap:dchow-vs-dh}

As mentioned in the body of the paper, our proof is essentially a refined
version of Theorem~4 of \cite{Goldberg:06b} with two
main modifications:  one is that we generalize Goldberg's arguments to
allow $g$ to be a bounded function rather than a Boolean
function, and the other is that we get rid of various factors of
$\sqrt{n}$ which arise in the \cite{Goldberg:06b} argument (and which would
be prohibitively ``expensive'' for us).
The key to getting rid of these
factors is the following
simple lemma:
%\rnote{Added a bit of exposition above to highlight the delta in
%what we are doing with what Goldberg did; let me know what you think}
\begin{lemma}\label{clm:3}
Let $S \subseteq \{-1,1\}^n$ and $\mcwa: S \rightarrow [0,2]$ such that $\littlesum_{x \in S} \mcwa(x) = \delta 2^n$.  Also,
let $v \in \R^n$ have $\| v \| =1$.  Then
$$\littlesum_{x \in S} \mcwa(x) \cdot |v \cdot x|  = O (\delta \sqrt{\log(1/\delta)} ) \cdot 2^n.$$
\end{lemma}

\begin{proof}
For any $x \in S$, let $D(x) \eqdef \mcwa(x)/(\littlesum_{x \in S} \mcwa(x))$. Clearly, $D$ defines a probability distribution over $S$.
By definition, $\E_{x \sim D}[|v \cdot x|] = (\littlesum_{x \in S} \mcwa(x) \cdot |v \cdot x|)/(\littlesum_{x \in S} \mcwa(x))$.
Since $\littlesum_{x \in S} \mcwa(x) = \delta \cdot 2^n$, to prove the lemma it suffices to show that $\E_{x \sim D}[|v \cdot x|] = O(\sqrt{\log(1/\delta)}).$
Recall that for any non-negative random variable $Y$, we have the identity $\E[Y] = \int_{t \ge 0} \Pr [Y>t] \ dt$. Thus, we have
$$\E_{x \sim D} [|v \cdot x|] = \int_{t \ge 0} \Pr_{x \sim D} [|v \cdot x|>t]  \ dt.$$
To bound this quantity, we exploit the fact that the integrand
is concentrated. Indeed, by the Hoeffding bound we have that $$\Pr_{x \sim \bn} [ |v \cdot x | > t ] \leq 2 e^{-t^2/2}.$$  This implies that the set $A= \{ x \in \bn :  |v \cdot x | > t \}$ is of size at most $2 e^{-t^2/2} 2^n$.
Since $\mcwa(x) \le 2$ for all $x \in S$, we have that
$\littlesum_{x \in A \cap S} \mcwa(x) \le 4 e^{-t^2/2} 2^n$. This implies that $\Pr_{x \sim D} [ |v \cdot x | > t ] \leq (4/\delta) \cdot e^{-t^2/2}$.
The following chain of inequalities completes the proof:
%$$\E_{x \sim D} \left[ |v \cdot x| \right]  = $$
%$$\int_{t = 0}^{\sqrt{2\ln(1/\delta)}} \Pr_{x \sim D} [|w \cdot x|>t]  \ dt + \int_{t \ge\sqrt{2\ln(1/\delta)}} \Pr_{x \sim D} [|v \cdot x|>t]  \ dt$$
\begin{eqnarray*}
\E_{x \sim D} \left[ |v \cdot x| \right]  &=& \int_{t = 0}^{\sqrt{2\ln(1/\delta)}} \Pr_{x \sim D} [|w \cdot x|>t]  \ dt + \int_{t \ge\sqrt{2\ln(1/\delta)}} \Pr_{x \sim D} [|v \cdot x|>t]  \ dt  \\
&\le& \sqrt{2\ln(1/\delta)} + \int_{t \ge\sqrt{2\ln(1/\delta)}} \Pr_{x \sim D} [|v \cdot x|>t]  \ dt \\
&\le&  \sqrt{2\ln(1/\delta)} + \int_{t \ge\sqrt{2\ln(1/\delta)}} \frac{4 e^{-t^2/2}}{\delta} \ dt \\
&\le& \sqrt{2\ln(1/\delta)} + \int_{t \ge\sqrt{2\ln(1/\delta)}} \frac{4 te^{-t^2/2}}{\delta} \ dt = \sqrt{2\ln(1/\delta)}  + 4.
\end{eqnarray*}
%
%Therefore, if $t>\sqrt{\ln(2/\delta)}$ we have $\Pr_{x \in S} [ | w\cdot  x | > t ] \le e^{-t^2/2}/\delta$.  We following chain of inequalities completes the proof:
%\begin{eqnarray*}
%\E_{x \in S} \left[ |w \cdot x| \right]  &=& \int_{t \ge 0}^{\sqrt{\ln(2/\delta)}} \Pr_{x \in S} [|w \cdot x|>t]  \ dt  + \int_{t \ge\sqrt{\ln(2/\delta)}} \Pr_{x \in S} [|w \cdot x|>t]  \ dt  \\ &\le& \sqrt{\ln(2/\delta)} + \int_{t \ge\sqrt{\ln(2/\delta)}} \Pr_{x \in S} [|w \cdot x|>t]  \ dt \\ &\le&  \sqrt{\ln(2/\delta)} + \int_{t \ge\sqrt{\ln(2/\delta)}} e^{-t^2/2}/\delta \ dt \\ &\le& \sqrt{\ln(2/\delta)} + o(1)
%\end{eqnarray*}
\end{proof}

We are now ready to prove Theorem~\ref{thm:dchow-vs-dh}.

\begin{proof}[Proof of Theorem~\ref{thm:dchow-vs-dh}]
Let $f:\bn \to \bits$ be an LTF and $g:\bn \to [-1,1]$ be an arbitrary
bounded function. Assuming that $\dist(f,g) = \eps$, we will prove
that $\dchow(f,g) \geq \delta = \delta(\eps) \eqdef \eps^{\Theta(\log^2(1/\eps))}$.

\smallskip

Let us define $V_{+} = \{x \in \bn \mid f(x) = 1, g(x) < 1\}$ and $V_{-} = \{x \in \{-1,1\}^n \mid f(x) =- 1, g(x) >-1\}$.
Also, for every point $x \in \bn$, we associate a weight $\mcwa(x) = |f(x) - g(x)|$ and for a set $S$, we define $\mcwa(S) \eqdef \sum_{x \in S} \mcwa(x)$.

It is clear that $V_{+} \cup V_{-}$ is the disagreement region
between $f$ and $g$ and that therefore $\mcwa(V_+) + \mcwa(V_-) = \eps \cdot 2^n$. We claim that without loss of generality we may assume that
$(\eps-\delta) \cdot 2^{n-1} \leq \mcwa(V_+), \mcwa(V_-) \leq (\eps+\delta) \cdot 2^{n-1}$. Indeed, if this condition is not satisfied,
we have that $|\wh{f}(0) - \wh{g}(0)| >\delta$ which gives the
conclusion of the theorem.

We record the following straightforward fact which shall be used several times subsequently.
\begin{fact}\label{fac:large}
For $\mcwa$ as defined above, for all $X \subseteq \{-1,1\}^n$, $|X| \ge \mcwa(X)/2$.
\end{fact}
\smallskip

We start by defining $V_+^0 = V_+$, $V_-^0 = V_-$ and $V^0 =V_+^0 \cup  V_-^0$.
The following simple proposition will be useful throughout the proof, since it characterizes the Chow distance between $f$ and $g$
(excluding the degree-$0$ coefficients) as the (normalized) Euclidean distance
between two well-defined points in $\R^n$:

\begin{proposition} \label{prop:euclidean}
Let $\mu_+ = \littlesum_{x \in V_+} \mcwa(x) \cdot x$ and $\mu_-
= \littlesum_{x \in V_-} \mcwa(x) \cdot x$. Then
$ \littlesum_{i=1}^n (\wh{f}(i) - \wh{g}(i))^2  = 2^{-2n} \cdot \| \mu_+ - \mu_- \|^2.$
\end{proposition}
\begin{proof}
For $ i \in [n]$ we have that $\wh{f}(i) = \E[f(x) x_i]$
and hence $\wh{f}(i) - \wh{g}(i) = \E[(f(x) - g(x)) x_i]$.
Hence $2^n (\wh{f}(i) - \wh{g}(i))  = \sum_{x \in V_+} \mcwa(x) \cdot x_i   - \sum_{x \in V_-} \mcwa(x) \cdot x_i = (\mu_+  - \mu_-) \cdot \mathbf{e}_i$ where $(\mu_+  - \mu_-) \cdot \mathbf{e}_i$ is the inner product of the vector $\mu_+  - \mu_-$ with the unit vector $\mathbf{e}_i$. Since $\mathbf{e}_1, \ldots, \mathbf{e}_n$ form a complete orthonormal basis for $\mathbb{R}^n$, it follows that
$$
\Vert \mu_+  - \mu_- \Vert^2 = 2^{2n}\sum_{i \in [n]}  (\wh{f}(i) - \wh{g}(i))^2
$$
proving the claim.
\end{proof}
%\begin{proof}
%TO DO.
%\end{proof}

%Note that with this definition, $\mu_+ - \mu_-$ is simply a point in $\mathbb{R}^n$.
If $\eta \in \R^n$ has $\| \eta \|=1$ then it is clear that $\| \mu_+ - \mu_- \| \geq  (\mu_+ - \mu_-) \cdot \eta$.
By Proposition~\ref{prop:euclidean}, to lower
bound the Chow distance $\dchow(f,g)$,
it suffices to establish a lower bound on $(\mu_+ - \mu_-) \cdot \eta$
for a unit vector $\eta$ of our choice.

Before proceeding with the proof we fix some notation. For any line $\ell$ in $\mathbb{R}^n$ and point $x \in \mathbb{R}^n$, we let $\ell(x)$ denote the projection of the point $x$ on the line $\ell$.  For a set $X \subseteq \mathbb{R}^n$ and a line $\ell$ in $\mathbb{R}^n$, $\ell(X) \eqdef \{\ell(x) : x \in X\}$. We use $\wh{\ell}$ to denote the unit vector in the direction of $\ell$ (its orientation is irrelevant for us).
\begin{definition}
For a function $\mcwa : \{-1,1\}^n \rightarrow [0,\infty)$,  a set $X \subseteq \{-1,1\}^n$  is  said to be \emph{$(\epsilon,\nu)$-balanced} if $(\epsilon - \nu) 2^{n-1} \le \littlesum_{x \in X} \mcwa(x) \le (\epsilon + \nu) 2^{n-1}$.
\end{definition}

%\medskip
\ignore{Unless explicitly stated otherwise,}Whenever we say that a set $X$ is $(\epsilon,\nu)$-balanced, the associated function $\mcwa$ is implicitly assumed to be the one defined at the start of the proof of Theorem~\ref{thm:dchow-vs-dh}.  The following proposition will be very useful during the course of the proof.
\begin{proposition}\label{prop:separation}
Let $X_1, X_2 \subseteq \{-1,1\}^n$ be $(\epsilon,\nu)$-balanced sets where $\nu \le \epsilon/8$.  Let  $\ell$ be a line in $\mathbb{R}^n$ and $q \in \ell$ be a point on $\ell$ such that
the sets $\ell(X_1)$ and $\ell(X_2)$ lie on opposite sides of $q$.  Suppose that $ S \eqdef \{x \mid x \in X_1 \cup X_2 \textrm { and } \Vert \ell(x) - q \Vert \ge \beta \}$.  If $\littlesum_{x \in S} \mcwa(x) \ge \gamma 2^n$, then for $\mu_1 = \sum_{x \in X_1} \mcwa(x) \cdot x$ and $\mu_2 = \sum_{x \in X_2} \mcwa(x) \cdot x$, we have $$|(\mu_1 - \mu_2) \cdot \wh{\ell} | \ge  (\beta \gamma   - \nu \sqrt{ 2\ln(16/\epsilon)} ) 2^n. $$
In particular, for $\nu \sqrt{2 \ln(16/\epsilon)} \le \beta \gamma/2$, we have $|(\mu_1 - \mu_2) \cdot \wh{\ell} | \ge (\beta \gamma /2) 2^n$.
\end{proposition}
\begin{proof}
We may assume that the projection $\ell(x)$ of any point $x \in X_1$ on $\ell$ is of the form $q + \lambda_x \wh{\ell}$ where $\lambda_x >0$, and that the projection $\ell(x)$ of any point $x \in X_2$ on $\ell$ is of the form $q - \lambda_x \wh{\ell}$ where $\lambda_x >0$. We can thus write \begin{eqnarray*}
(\mu_1 - \mu_2) \cdot \wh{\ell}  &=& \littlesum_{ x \in X_1} \mcwa(x) (q \cdot \wh{\ell} + \lambda_x)   -  \littlesum_{ x \in X_{2}} \mcwa(x) (q \cdot \wh{\ell} - \lambda_x) \\
&=&   \left(  \mcwa(X_1) - \mcwa(X_2)\right)q \cdot \wh{\ell}   + \littlesum_{x \in X_1 \cup X_2} \mcwa(x) \cdot \lambda_x.
\end{eqnarray*}
By the triangle inequality we have
$$ \left| (\mu_1 - \mu_2) \cdot \wh{\ell} \right|  \geq   \littlesum_{x \in X_1 \cup X_2} \mcwa(x) \cdot \lambda_x  -  |q \cdot \wh{\ell}|  \left|(  \mcwa(X_1) - \mcwa(X_2) )\right|  $$
so it suffices to bound each term separately.
For the first term we can write $$\littlesum_{x \in X_1 \cup X_2}  \mcwa(x) \cdot \lambda_x \ge \littlesum_{x \in S} \mcwa(x) \cdot \lambda_x \geq  \beta \gamma 2^n.$$
To bound the second term, we first recall that (by assumption) $ \left| \mcwa(X_1) - \mcwa(X_2) \right| \leq \nu 2^n$. Also, we claim that $ |q \cdot \wh{\ell}|< \sqrt{2\ln(16/\epsilon)}$. This is because otherwise the function defined by $g(x) = \sign(x \cdot \wh{\ell}  - q \cdot \wh{\ell})$ will be $\epsilon/8 $ close to a constant function on $\{-1,1\}^n$.
In particular, at least one of $|X_1|, |X_2|$
must be at most $(\epsilon/8) 2^n$.
However, by Fact~\ref{fac:large}, for $i=1,2$ we have that
$|X_i| \ge \mcwa(X_i)/2 \ge (\epsilon/4 -\nu/4) 2^n > (\epsilon/8) 2^n $
resulting in a contradiction.
Hence it must be the case that
$ |q \cdot \wh{\ell}|< \sqrt{2\ln(16/\epsilon)}$.
This implies that  $|(\mu_1 - \mu_2) \cdot \wh{\ell} | \ge  (\beta \gamma   - \nu \sqrt{ 2\ln(16/\epsilon)} ) 2^n $ and the proposition is proved.
\end{proof}

 \medskip

We consider a separating hyperplane $\bA_0$ for $f$ and assume (without loss of generality)
that $\bA_0$ does not contain any points of the unit hypercube $\bn$.
Let $\bA_0 = \{ x \in \R^n \mid w \cdot x = \theta \}$, where $\| w \| = 1$, $\theta \in \R$ and
$f(x) = \sgn(w \cdot x - \theta)$.

Consider a line $\ell_0$ normal to $\bA_0$, so
$w$ is the unit vector defining the direction of $\ell_0$
that points to the halfspace $f^{-1}(1)$. \newad{As stated before, the exact orientation of $\ell_0$ is irrelevant to us and the choice of orientation here is arbitrary.} Let $q_0 \in \R^n$ be the intersection point of $\ell_0$ and $\bA_0$.
Then we can write the line $\ell_0$ as
$\ell_0 = \{ p \in \R^n \mid p = q_0 + \lambda w, \lambda \in \R  \}.$

%The point $q_0$ is naturally associated with a real number $Q_0$ as follows: Let $p_{\mathbf{0}}$ be the projection of the origin $\mathbf{0}_n \in \R^n$
%on $\ell_0$. Then $Q_0$ is basically the distance between $q_0$ and $p_{\mathbf{0}}$ with the appropriate sign. In particular, we define
%$Q_0 = \eta_0 \cdot \vec{q_0p_{\mathbf{0}}}$.

Define $\beta \eqdef \epsilon^{O(\log(1/\epsilon))}$ and consider the set of points
$$
S_0 = \{x : x \in V^0 \mid \| \ell_0(x) - q_0\| \ge \beta \}.
$$

The following claim states that if $\mcwa(S_0)$  is not very small,
we get the desired lower bound on the Chow distance.

\begin{claim}\label{clm:4}
Suppose that $\mcwa(S_0) \ge \gamma_0 \cdot 2^n$ where $\gamma_0 \eqdef \beta^{4 \log(1/\epsilon) -2} \cdot \epsilon$. Then $\dchow(f,g) \ge \delta$.
\end{claim}

\begin{proof}
To prove the desired lower bound, we will apply
Proposition~\ref{prop:euclidean}. Consider projecting every point in $V^0$
on the line $\ell_0$.  Observe that the projections of $V_{+}^{0}$ are
separated from the projections of $V_{-}^{0}$ by the point $q_0$. Also, we
recall that the sets $V_+^0$ and $V_-^0$ are $(\epsilon, \delta)$ balanced.
Thus, if we define $\mu_+ = \littlesum_{x \in V_+^0} \mcwa(x) \cdot x$
and $\mu_- = \littlesum_{x \in V_-^0} \mcwa(x) \cdot x$, we can apply
Proposition~\ref{prop:separation} to get that $|(\mu_+ - \mu_-) \cdot w|
\ge (\beta \gamma_0 - \delta \sqrt{2\ln(16/\eps)} ) 2^n \ge \delta 2^n$.
This implies that $\Vert \mu_+ - \mu_- \Vert^2 \ge \delta^2 2^{2n}$ and
using Proposition~\ref{prop:euclidean}, this proves that $\dchow(f,g)
\ge \delta$.
\end{proof}

If the condition of Claim~\ref{clm:4} is not satisfied,
then we have that $\mcwa(V^0 \setminus S_0) \ge (\epsilon - \gamma_0) 2^n$.  By Fact~\ref{fac:large}, we
have $|V^0 \setminus S_0| \ge (\epsilon -\gamma_0)2^{n-1}$.
We now apply
Lemma~\ref{lem:1} to obtain another hyperplane $\bA_1$ which passes through
all but $\kappa_1 \cdot2^n$  points ($\kappa_1 \eqdef \gamma_0/2$) in $V^0 \setminus S_0$.
We note that the condition of the lemma is satisfied, as $\log(1/\kappa_1) = \poly(\log(1/\epsilon))$
and $|V^0 \setminus S_0| > (\epsilon/4) \cdot 2^n$.

From this point onwards, our proof uses a sequence of $\lfloor \log(1/\epsilon) \rfloor$ cases.   To this end, we define $\gamma_j = \beta^{4 \log (1/\epsilon) -2(j+1)} \cdot \epsilon$. At the beginning of case $j$, we will have an
affine space $A_j$ of dimension $n-j$ such that $\mcwa(V^0 \cap A_j)
\ge (\epsilon-2(\littlesum_{\ell=0}^{j-1} \gamma_{\ell})) 2^n $. We note
that this is indeed satisfied at the beginning of case $1$. To see this,
recall that  $\mcwa(V^0 \setminus S_0) >(\epsilon - \gamma_0) 2^n$. Also,
we have that
\begin{eqnarray*}
W((V^0 \setminus S_0) \setminus (V^0 \cap \bA_1)) &\le& 2|(V^0\setminus S_0) \setminus (V^0 \cap \bA_1)| \\
&\le& 2 \kappa_1 2^n = \gamma_0 2^n.
\end{eqnarray*}
These together imply that $\mcwa(V^0 \cap \bA_1) \ge (\epsilon-2\gamma_0) 2^n $
confirming the hypothesis for  $j=1$.

\smallskip

We next define $V^j = V^0 \cap A_j$, $V^j_{+} = V^j \cap V_{+}$ and
$V^j_{-} = V^j \cap V_-$.  Similarly, define $\Delta^j_+  =  V^0_{+} \setminus V^j_{+} $ and $\Delta^j_-  =  V^0_{-} \setminus V^j_{-} $. Let $A'_{j+1} = A_j \cap \bA_0$.  Note that $A_j \not \subseteq \bA_0$. This is because $A_j$ contains points from $\{-1,1\}^n$ as opposed to $\bA_0$ which does not. Also, $A_j$ is not contained in a hyperplane parallel to $\bA_0$ because $A_j$ contains points of the unit hypercube lying on either side of $\bA_0$. Hence
it must be the case that
$\dim(A'_{j+1}) = n -(j+1)$. Let $\ell_j$ be a line orthogonal to $A'_{j+1}$ which is parallel to  $A_j$. Again, we observe  that the direction of $\ell_j$ is unique.

We next observe that all points in $A'_{j+1}$ project to the same
point in $\ell_j$, which we call $q_j$.  Let us
define  $\Lambda^j_+ = \ell_j(V^j_+)$ and $\Lambda^j_- = \ell_j(V^j_-)$.  We state the following important observation.
\begin{obs}\label{obs:1}
The sets $\Lambda^j_+$ and $\Lambda^j_-$ are separated by $q_j$.
\end{obs}
Next, we define $S_j$ as :
$$
S_ j  =  \{x \in V^j \mid \| \ell_j(x) - q_j\|_2 \ge \beta \}.
$$
The next claim  is analogous to Claim~\ref{clm:4}. It says that
if $\mcwa(S_j)$ is not too small, then we get the desired
lower bound on the Chow distance. \newad{The proof is slightly more technical and uses Lemma~\ref{clm:3}}.
\begin{claim}\label{clm:induction}
For $j \le \log(8/\epsilon)$, suppose that $\mcwa(S_j)\ge \gamma_j \cdot 2^n$
where $\gamma_j$ is as defined above. Then $\dchow(f,g) \ge \delta$.
\end{claim}
\begin{proof}
We start by observing that  $$\left( \epsilon  - 4\littlesum_{\ell=0}^{j-1} \gamma_{\ell} \right)2^{n-1} \le \mcwa(V^j_+), \mcwa(V^j_-) \le (\epsilon + \delta) 2^{n-1}.$$ The upper bound is obvious because $V_+^j \subseteq V_+^0$ and $V_-^j \subseteq V_-^0$ and the range of $\mcwa$ is non-negative. To see the lower bound, note that $\mcwa(V^0 \setminus V^j) \le 2(\littlesum_{\ell=0}^{j-1} \gamma_{\ell})2^{n}$.  As $V^0_+ \setminus V^j_+$ and $V^0_- \setminus V^j_-$ are both contained in $V^0 \setminus V^j$, we get the stated lower bound. We also note that
\begin{eqnarray*}
2\left(\sum_{\ell=0}^{j-1} \gamma_{\ell}\right)2^n  &=& 2 \left(\sum_{\ell=0}^{j-1} \beta^{4 \log(1/\epsilon) -2 \ell -2} \right) 2^n \\
                                                                                        &\le& 4 \beta^{4 \log(1/\epsilon) -2j}  2^n.
\end{eqnarray*}
This implies that the sets $V^j_+$ and $V^j_-$ are $(\epsilon, 4 \beta^{4 \log(1/\epsilon) -2j}  + \delta)$ balanced.  In particular, using that $\delta \le 4 \beta^{4 \log(1/\epsilon) -2j}$,  we can say that the sets $V^j_+$ and $V^j_-$ are $(\epsilon,8 \beta^{4 \log(1/\epsilon) -2j})$-balanced.  We also observe that for $j \le \log(8/\epsilon)$, we have that $8 \beta^{4 \log(1/\epsilon) -2j} \le \epsilon/8$.   Let us define $\mu_+^j = \sum_{x \in V_+^j} \mcwa(x) \cdot x$ and $\mu_-^j = \sum_{x \in V_-^j} \mcwa(x) \cdot x$.
An application of Proposition~\ref{prop:separation} yields that  $|(\mu_+^j - \mu_-^j) \cdot \wh{\ell_j}| \ge (\beta \gamma_j - 8 \beta^{4 \log(1/\epsilon) -2j} \sqrt{2 \ln(16/\epsilon)})2^n$.

\smallskip

We now note that
$$(\mu_+ - \mu_-)\cdot \wh{\ell_j}  = (\mu_+^j - \mu_-^j) \cdot \wh{\ell_j} + \left( \sum_{x \in \Delta^j_+} \mcwa(x) - \sum_{x \in \Delta^j_-} \mcwa(x) \right) \cdot \wh{\ell_j}.$$ Defining $\mu'^j_+ =\littlesum_{x \in \Delta^j_+} \mcwa(x) \cdot x$ and $\mu'^j_- =\littlesum_{x \in \Delta^j_-} \mcwa(x) \cdot x$, the triangle inequality implies that $$
\left|\left(\mu_+ - \mu_-\right)\cdot \wh{\ell_j} \right| \ge \left|\left(\mu_+^j - \mu_-^j\right) \cdot \wh{\ell_j}\right| - \left|\mu'^j_+  \cdot \wh{\ell_j}\right| -\left| \mu'^j_-\cdot \wh{\ell_j} \right|.
$$
Using Lemma~\ref{clm:3} and that $\mcwa(\Delta^j_+), \mcwa(\Delta^j_-)\le  \mcwa(V^0 \setminus V^j) \le 8 \beta^{4 \log(1/\epsilon) -2j}\cdot 2^n$, we get that
\begin{eqnarray*}
\left|\mu'^j_+  \cdot \wh{\ell_j}\right| &=& \littlesum_{x \in \Delta^j_+} \mcwa(x) \cdot x \cdot \wh{\ell_j} \\
                                                               &=& O\left(|\Delta^j_+| \cdot \sqrt{\log(2^n/|\Delta^j_+|) } \right)\\
                                                               &=& O\left(  \beta^{4 \log(1/\epsilon) -2j} \cdot \log^{3/2}(1/\epsilon) \cdot 2^n \right)
\end{eqnarray*}
and similarly
\begin{eqnarray*}
\left|\mu'^j_-  \cdot \wh{\ell_j}\right| &=& \littlesum_{x \in \Delta^j_-} \mcwa(x) \cdot x)\cdot \wh{\ell_j} \\
                                                              &=& O\left(|\Delta^j_-| \cdot \sqrt{\log(2^n/|\Delta^j_-|) } \right)\\
                                                             &=& O\left(  \beta^{4 \log(1/\epsilon) -2j} \cdot \log^{3/2}(1/\epsilon) \cdot 2^n \right).
\end{eqnarray*}
This implies that
\begin{eqnarray*}
\left|\left(\mu_+ - \mu_-\right)\cdot \wh{\ell_j} \right| \ge (\beta \gamma_j - 8 \beta^{4 \log(1/\epsilon) -2j} \sqrt{2 \ln(8/\epsilon)})2^n  \\ -
O\left(  \beta^{4 \log(1/\epsilon) -2j} \cdot \log^{3/2}(1/\epsilon) \cdot 2^n \right).
\end{eqnarray*}
Plugging in the value of $\gamma_j$, we see  that for $\eps$ smaller than a sufficiently small constant, we have that
 $$\left|\left(\mu_+ - \mu_-\right)\cdot \wh{\ell_j} \right| \ge \beta \gamma_j 2^{n-1}. $$
 An application of Proposition~\ref{prop:euclidean} finally gives us that
 $$
 \dchow(f,g)  \ge 2^{-n} \Vert \mu_+ - \mu_- \Vert  \ge 2^{-n} (\mu_+ - \mu_- ) \cdot \wh{\ell_j} =\beta \gamma_j /2  \ge \delta
 $$
 which establishes the Claim.
\end{proof}
\medskip
If the hypothesis of Claim~\ref{clm:induction} fails, then we construct an
affine space $A_{j+1}$ of dimension $n-j-1$ such that $\mcwa(V^0
\cap A_{j+1}) \ge (\epsilon - 2 \littlesum_{\ell=0}^j \gamma_{\ell})2^n$
as described next.
We recall that $U = \cup_{i=1}^n \mathbf{e}_i \cup \mathbf{0}$.
It is obvious there is some subset $Y_j  \subseteq U$ such that
$|Y_j|=j$  and $\span(A_j \cup Y_j) = \mathbb{R}^n$.  Now, let us
define $\h'_j \eqdef \span(Y_j \cup A'_{j+1})$. Clearly, $\h'_j$ is a
hyperplane and every point $x \in (V^0 \cap A_j) \setminus S_j$ is at a distance at most
$\beta$ from $H'_{j}$. This is because every $x \in (V^0 \cap A_j) \setminus S_j$ is at a
distance at most $\beta$ from $A'_{j+1}$ and $A'_{j+1} \subset \h'_j$. Also,
note that all $x \in Y_j$ lie on $\h'_j$.

Note that $\mcwa((V^0 \cap A_j) \setminus S_j) \ge (\epsilon - 2 \littlesum_{\ell=0}^{j-1} \gamma_{\ell} - \gamma_j)2^n$. As prior calculation has shown, for $j \le  \log(8/\epsilon)$ we have
$
\mcwa((V^0 \cap A_j) \setminus S_j) \ge (\epsilon - 2 \littlesum_{\ell=0}^{j-1} \gamma_{\ell} - \gamma_j)2^n \ge (\epsilon/2) 2^n$.  Using Fact~\ref{fac:large}, we get that $|(V^0 \cap A_j) \setminus S_j| \ge (\epsilon/4) 2^n$.  Thus, putting $\kappa_{j} = \gamma_j/2$ and applying Lemma~\ref{lem:1}, we get a new hyperplane $\h_j$ such that $|  ((V^0 \cap A_j) \setminus S_j) \setminus (\h_j \cap V^0)| \le (\gamma_j/2) \cdot2^n$. Using that the range of $\mcwa$ is bounded by $2$, we get
$\mcwa( ((V^0 \cap A_j) \setminus S_j) \setminus (\h_j \cap V^0)) \le  \gamma_j \cdot2^n$.   Thus, we get that  $\mcwa(\h_j \cap V^0 \cap A_j) \ge (\epsilon - 2 \littlesum_{\ell=0}^{j} \gamma_{\ell})2^n$.  Also, $Y_j \subset \h_j$.

Let us now define $A_{j+1} = A_j \cap \h_j$. It is clear that  $\mcwa(A_{j+1} \cap V^0) \ge (\epsilon - 2 \littlesum_{\ell=0}^{j} \gamma_{\ell})2^n$. Also, $\dim(A_{j+1}) < \dim(A_j)$. To see this, assume for contradiction that $\dim(A_j) = \dim(A_{j+1})$. This means that $A_j \subseteq \h_j$. Also, $Y_j \subset \h_j$. This means that $\span(A_j \cup Y_j) \subset \h_j$. But $\span(A_j \cup Y_j) = \mathbb{R}^n$ which cannot be contained in $\h_j.$ Thus we have that $\dim(A_{j+1}) = \dim(A_j) -1$.

Now we observe that taking $j=\lfloor \log(8/\epsilon) \rfloor$, we have a subspace $A_j$ of dimension $n-j$ which has $\mcwa(A_j \cap V^0) \ge (\epsilon - 2 \littlesum_{\ell=0}^{j-1} \gamma_{\ell}) 2^n > (\epsilon/2)2^n$. By Fact~\ref{fac:large}, we have that  $|A_j \cap V^0| \ge (\epsilon/4) 2^n$.  However, by Fact~\ref{fac:affine}, a subspace of dimension $n-j$ can contain at most $2^{n-j}$ points of $\{-1,1\}^n$.  Since $j =\lfloor \log(8/\epsilon) \rfloor$, this leads to a contradiction.  That implies that the number of cases must be strictly less than $\lfloor \log(8/\epsilon) \rfloor$. In particular, for some $j < \lfloor \log(8/\epsilon) \rfloor$,
it must be the case that $|S_j| \ge \gamma_j 2^n $. For this $j$, by Claim~\ref{clm:induction}, we get a lower bound of $\delta$ on $\dchow(f,g)$.
This concludes the proof of Theorem~\ref{thm:dchow-vs-dh}.
\end{proof}

\section{The Algorithm and its Analysis} \label{sec:chow-algo}

\subsection{Algorithm and Proof Overview.} \label{ssec:algo}

In this section we give a proof overview of Theorem~\ref{thm:alg}, restated below for convenience.
We give the formal details of the proof in the following subsection.

\medskip

\noindent {\bf Theorem~\ref{thm:alg}} (Main Algorithmic Result). \emph{
There exists a randomized algorithm {\tt ChowReconstruct} that for every Boolean function $f: \{-1,1\}^n \rightarrow \{-1,1\}$,
given $\eps>0,\delta>0$ and a vector $\vec{\alpha} = (\alpha_0,\alpha_1,\ldots,\alpha_n)$ such that $\|\mychows_f - \vec{\alpha}\| \leq \eps$,
with probability at least $1-\delta$, outputs an LBF $g$ such that $\|\mychows_f - \mychows_g\| \leq 6 \eps$.
The algorithm runs in time $\tilde{O}(n^2 \eps^{-4} \log{(1/\delta}))$.
Further, $g$ is represented by a weight vector $\kappa v \in \R^{n+1}$, where $\kappa \in \R$ and $v$ is an integer vector of length $\|v\| = O(\sqrt{n}/\eps^3)$.
}
\medskip

We now provide an intuitive overview of the algorithm and its analysis.
Our algorithm is motivated by the following intuitive reasoning:  since the function $\alpha_0 + \sum_{i\in[n]} \alpha_i \cdot x_i$ has the desired Chow parameters, why not just use it to define an LBF $g_1$ as $P_1(\alpha_0 + \sum_{i\in[n]} \alpha_i \cdot x_i)$?  The answer, of course, is that as a result of applying the projection operator, the Chow parameters of $g_1$ can become quite different from the desired vector $\vec{\alpha}$.  Nevertheless, it seems quite plausible to expect that $g_1$ will be better than a random guess.

Given the Chow parameters of $g_1$ we can try to correct them by adding the difference between $\vec{\alpha}$ and $\vec{\chi}_{g_1}$ to the vector that represents $g_1$. Again, intuitively we are adding a real-valued function $h_1 = \alpha_0 - \wh{g}_1(0) + \sum_{i\in[n]} (\alpha_i - \wh{g}_1(i)) \cdot x_i$ with the Chow parameters that we would like to add to the Chow parameters of $g_1$. And, again, the projection operation is likely to ruin our intention but we could still hope that we got closer to $f$ and that by doing this operation for a while we will converge to an LBF with Chow parameters close to $\vec{\alpha}$.

While this idea might appear too naive, this is almost exactly what we do in {\tt ChowReconstruct}.
The main difference between this naive proposal and our actual algorithm is that at step $t$ we actually add only half the difference
between $\vec{\alpha}$ and the Chow vector of the current hypothesis $\vec{\chi}_{g_t}$. This is necessary in our proof to offset the fact that $\vec{\alpha}$ is only an approximation to $\vec{\chi}_{f}$ and we can only approximate the Chow parameters of $g_t$. An additional minor modification is required to ensure that the resulting weight vector is a multiple of an integer weight vector of length $O(\sqrt{n}/\eps^3)$.

Proving the correctness of this algorithm roughly proceeds as follows. If the difference vector is sufficiently large (namely, more than a small multiple of the difference between $\|\vec{\chi}_f - \vec{\alpha}\|)$ then the linear function $h_t$ defined by this vector can be easily seen as being correlated with $f-g_t$, namely $\E[(f-g_t)h_t] \geq c \|\vec{\chi}_{g_t} - \vec{\alpha}\|^2$ for a constant $c>0$. As was shown in \cite{TTV09} and \cite{Feldman:10ab} this condition for a Boolean $h_t$ can be used to decrease a simple potential function measuring $\E[(f-g_t)^2]$, the $l_2^2$ distance of the current hypothesis to $f$. One issue that arises is this: while the $l_2^2$ distance is only reduced if $h_t$ is added to $g_t$, in order to ensure that $g_{t+1}$ is an LBF, we need to add the vector of difference (used to define $h_t$) to the weight vector representing $g_t$. To overcome this problem the proof in \cite{TTV09} uses an additional point-wise counting argument from \cite{imp95a}. This counting argument can be adapted to the real valued $h_t$, but the resulting argument becomes quite cumbersome. Instead, we augment the potential function in a way that captures the additional counting argument from \cite{imp95a} and easily generalizes to the real-valued case. % Deriving the necessary potential function follows the strategy of potential-based convergence arguments in \cite{Feldman:10ab}.

\subsection{Proof of Theorem~\ref{thm:alg}.} \label{ap:alg}
We build $g$ through the following iterative process.
Let $g'_0 \equiv 0$ and let $g_0 = P_1(g'_0)$. Given $g_t$, we compute the Chow parameters of $g_t$ to accuracy $\eps/(4\sqrt{n+1})$ and let $(\beta_0,\beta_1,\ldots,\beta_n)$ denote the results. For each $0\leq i\leq n$ we define $\tilde{g}_t(i)$ to be the closest value to $\beta_i$ that ensures that $\alpha_i - \beta_i$ is an integer multiple of $\eps/(2\sqrt{n+1})$. Let $\tilde{\chi}_{g_t} = (\tilde{g}_t(0),\ldots, \tilde{g}_t(n))$ denote the resulting vector of coefficients. Note that $$\|\tilde{\chi}_{g_t}-\vec{\chi}_{g_t}\| \leq \sqrt{\sum_{i=0}^n (\eps/(2\sqrt{n+1}))^2} = \eps/2 .$$

If $\rho \triangleq \|\vec{\alpha}-\tilde{\chi}_{g_t}\| \leq 4 \eps$ then we stop and output $g_t$.
By triangle inequality,
\begin{eqnarray*}
\| \vec{\chi}_f - \vec{\chi}_{g_t}\| &\leq& \| \vec{\chi}_f - \vec{\alpha}\|  + \| \vec{\alpha} - \tilde{\chi}_{g_t}\| +  \|\tilde{\chi}_{g_t}-\vec{\chi}_{g_t}\| \\
                                                       &\leq& \eps (1+ 4 + 1/2) < 6 \eps,
\end{eqnarray*}
 in other words $g_t$ satisfies the claimed condition.

Otherwise (when $\rho > 4\eps$), let $g'_{t+1} = g'_t + h_t/2$ and $g_{t+1} = P_1(g'_{t+1})$ for $$h_t \triangleq \sum_{i=0}^n (\alpha_i - \tilde{g_t}(i)) x_i.$$ Note that this is equivalent to adding the vector $(\vec{\alpha} - \tilde{\chi}_{g_t})/2$ to the degree 0 and 1 Fourier coefficients of $g'_{t}$ (which are also the components of the vector representing $g_t$).

To prove the convergence of this process we define a potential function at step $t$ as
\begin{eqnarray*}
E(t) &=& \E[(f-g_t)^2] + 2\E[(f-g_t)(g_t-g'_t)] \\
       &=&  \E[(f-g_t)(f-2g'_t+g_t)].
\end{eqnarray*}
The key claim of this proof is that
$$E(t+1) - E(t) \leq -2\eps^2.$$ To prove this claim we first prove that
\begin{equation}\label{eq:bound-gradient}
\E[(f-g_t) h_t ] \geq \rho(\rho - \frac{3}{2}\eps).
\end{equation}
To prove equation (\ref{eq:bound-gradient}) we observe that, by Cauchy-Schwartz inequality,
\begin{eqnarray*}
\E[(f-g_t) h_t ] &=& \sum_{i=0}^n (\wh{f}(i)-\wh{g_t}(i))(\alpha_i-\tilde{g_t}(i)) \\
&=& \sum_{i=0}^n \Big[ (\wh{f}(i)-\alpha_i)(\alpha_i-\tilde{g_t}(i)) + \\
&&(\tilde{g_t}(i)-\wh{g_t}(i))(\alpha_i-\tilde{g_t}(i)) + (\alpha_i-\tilde{g_t}(i))^2 \Big]  \\
&\geq& -\rho \eps - \rho \eps/2 + \rho^2 \geq \rho^2 - \frac{3}{2}\rho\eps.
\end{eqnarray*}
In addition, by Parseval's identity,
\equ{\E[h_t^2] = \sum_{i=0}^n (\alpha_i - \tilde{g_t}(i))^2 = \rho^2 \ . \label{eq:bound-tail}}
Now,
\begin{eqnarray}
E(t+1) - E(t) &=& \E[(f-g_{t+1})(f-2g'_{t+1}+g_{t+1})] -\E[(f-g_t)(f-2g'_t+g_t)] \notag   \\
&=&  \E\left[(f-g_t)(2 g'_t - 2 g'_{t+1}) + (g_{t+1}-g_t)(2g'_{t+1}-g_t-g_{t+1})        \right] \notag \\
&=& - \E[(f-g_t)h_t]+ \E\left[(g_{t+1}-g_t)(2g'_{t+1}-g_t-g_{t+1})  \right] \label{eq:bound-step}
\end{eqnarray}

To upper-bound the expression $\E\left[(g_{t+1}-g_t)(2g'_{t+1}-g_t-g_{t+1})\right]$ we prove that for every point $x \in \{-1,1\}^n$, \equn{(g_{t+1}(x)-g_t(x))(2g'_{t+1}(x)-g_t(x)-g_{t+1}(x)) \leq h_t(x)^2/2.} We first observe that $$|g_{t+1}(x)-g_t(x)| = |P_1(g'_{t}(x)+ h_t(x)/2) - P_1(g'_t(x))| \leq |h_t(x)/2|$$ (a projection operation does not increase the distance).
Now
\begin{eqnarray*}
|2 g'_{t+1}(x)-g_t(x)-g_{t+1}(x)| \leq
 |g'_{t+1}(x)-g_t(x)| + |(g'_{t+1}(x)-g_{t+1}(x)|.
\end{eqnarray*}

The first part $|g'_{t+1}(x)-g_t(x)| = |h_t(x)/2 + g'_t(x)-g_t(x)| \leq |h_t(x)/2|$ unless $g'_t(x)- g_t(x) \neq 0$ and $g'_t(x)- g_t(x)$ has the same sign as $h_t(x)$. By the definition of $P_1$, this implies that $|g_t(x)| = \sgn(g'_t(x))$ and $\sgn(h_t(x)) = \sgn(g'_t(x) - g_t(x)) = g_t(x)$. However, in this case $|g'_{t+1}(x)| \geq |g'_{t}(x)| > 1$ and $\sgn(g'_{t+1}(x)) = \sgn(g'_t(x)) = g_t(x)$. As a result $g_{t+1}(x)=g_t(x)$ and $(g_{t+1}(x)-g_t(x))(2g'_{t+1}(x)-g_t(x)-g_{t+1}(x)) =0$. Similarly, for the second part: $|g'_{t+1}(x)-g_{t+1}(x)| > |h_t(x)/2|$ implies that $g_{t+1}(x) = \sgn(g'_{t+1}(x))$ and $|g'_{t+1}(x)| \geq |h_t(x)/2| + 1$. This implies that $|g'_t(x)| \geq |g'_{t+1}(x)| - |h_t(x)/2| > 1$ and $g_t(x) = \sgn(g'_t(x)) = \sgn(g'_{t+1}(x)) = g_{t+1}(x)$. Altogether we obtain that
\begin{eqnarray*}
(g_{t+1}(x)-g_t(x))(2g'_{t+1}(x)-g_t(x)-g_{t+1}(x)) \leq
\max\{0, |h_t(x)/2| ( |h_t(x)/2| + |h_t(x)/2|)\} = h_t(x)^2/2.
\end{eqnarray*}
This implies that \equ{\E\left[(g_{t+1}-g_t)(2g'_{t+1}-g_t-g_{t+1})\right] \leq \E[h_t^2]/2 = \rho^2/2 \label{eq:bound-gs}.}

By substituting equations (\ref{eq:bound-gradient}) and (\ref{eq:bound-gs}) into equation (\ref{eq:bound-step}), we obtain the claimed decrease in the potential function $$E(t+1) - E(t) \leq -\rho^2 + \frac{3}{2}\rho\eps + \rho^2/2 = -(\rho - 3\eps)\rho/2 \leq -2\eps^2.$$

We now observe that $$E(t)  = \E[(f-g_t)^2] + 2\E[(f-g_t)(g_t-g'_t)] \geq 0$$ for all $t$. This follows from noting that for every $x$ and $f(x) \in \{-1,1\}$, if $g_t(x)-g'_t(x)$ is non-zero then, by the definition of $P_1$, $g_t(x) = \sgn(g'_t(x))$ and $\sgn(g_t(x)-g'_t(x)) = -g_t(x)$. In this case, $f(x)-g_t(x)=0$ or $\sgn(f(x)-g_t(x)) = -g_t(x)$ and hence $(f(x)-g_t(x))(g_t(x)-g'_t(x)) \geq 0$. Therefore $$\E[(f-g_t)(g_t-g'_t)] \geq 0$$ (and, naturally, $\E[(f-g_t)^2] \geq 0$).  It is easy to see that $E(0) = 1$ and therefore this process will stop after at most $1/(2\eps^2)$ steps.

We now establish the claimed weight bound on the LBF output by the algorithm and the bound on the running time. Let $T$ denote the number of iterations of the algorithm. By our construction, the function $g_T = P_1(\sum_{t \leq T} h_t/2)$ is an LBF represented by weight vector $\vec{w}$ such that $w_i = \sum_{j\leq T} (\alpha_i - \tilde{g_j}(i))/2$. Our rounding of the estimates of Chow parameters of $g_t$ ensures that each of $(\alpha_i - \tilde{g_j}(i))/2$ is a multiple of  $\kappa = \eps/(4\sqrt{n+1})$. Hence $g_T$ can be  represented by vector $\vec{w} = \kappa \vec{v}$, where vector $\vec{v}$ has only integer components. At every step $j$, $$\sqrt{\sum_{i=0}^n (\alpha_i - \tilde{g_j}(i))^2} \leq 2+\eps+\eps/2 = O(1).$$ Therefore, by triangle inequality, $\|\vec{w}\| = O(\eps^{-2})$ and hence $\|\vec{v}\| = \|\vec{w}\|/\kappa = O(\sqrt{n}/\eps^3)$.

The running time of the algorithm is essentially determined by finding $\tilde{\chi}_{g_t}$ in each step $t$. Finding $\tilde{\chi}_{g_t}$ requires estimating each $\wh{g_t}(i) = \E[g_t(x) \cdot x_i]$ to accuracy $\eps/(4\sqrt{n+1})$. Chernoff bounds imply that, by using the empirical mean of $g_t(x) \cdot x_i$ on $O( (n/\eps^2) \cdot \log{(n/(\eps\delta))}$ random points as our estimate of $\wh{g_t}(i)$ we can ensure that, with probability at least $1-\delta$, the estimates are within $\eps/(4\sqrt{n+1})$ of the true values for all $n+1$ Chow parameters of $g_t$ for every $t \leq T = O(\eps^{-2})$.

Evaluating $g_t$ on any point $x \in \bn$ takes $O(n)$ time and we need to evaluate it on $O( (n/\eps^2) \cdot \log{(n/(\eps\delta))}$ points in each of $O(\eps^{-2})$ steps. This gives us the claimed total running time bound. %of $\tilde{O}(n^2 \eps^{-4}\log{(1/\delta)})$.

%\end{proof}

\ignore{
To convert an LBF to an LTF we rely on the fact that for an LBF $h$ defined by a vector $\vec{w}$ of weights $$L_1(h,f) = 2 \E_{\theta \in [-1,1]} \left[\pr[f \neq \sign(\vec{v}\vec{x} + \theta)]\right].$$ First approach: just choose a random threshold $\theta$ in $[-1,1]$ and output the hypothesis. By Markov inequality with probability at least $1/2$ we will obtain that
$\pr[f \neq \sign(\vec{w}\vec{x} + \theta)] \leq 4 \cdot L_1(h,f)$.

If $h$ can be represented using a vector of weights $\vec{w} = \kappa \vec{v}$  where each $v_i$ is an integer bounded by $W$ in absolute value then the only values of $\theta$ that need to be considered are those that are integer multiples of $\kappa$. We can multiply the weights
of the LPF $g$ obtained from the above procedure by $1/\kappa$ to obtain an equivalent integer weight representation of $g$. As a result each individual weights of $g$ will be upper-bounded by $W + \lceil 1/\kappa \rceil$ in absolute value and the $l_2$ norm of the weights by $|\vec{v}| + \lceil 1/\kappa \rceil$.

To boost the probability of success one could try many $\theta$'s and choose one with the closest Chow distance. That would require using the reverse relation between the Chow distance and ``compose" the relationship between Chow distance and Hamming distance. Should keep everything quasi-polynomial.
}

\section{The Main Results} \label{sec:payoff}

\subsection{Proofs of Theorems~\ref{thm:main} and~\ref{thm:lowwt}.} \label{ssec:main-proofs}

In this subsection we put the pieces together and prove our main results.
We start by giving a formal statement of Theorem~\ref{thm:main}:

\begin{theorem} [Main] \label{thm:main-formal}
There is a function $\chowallow(\eps) \eqdef 2^{-O(\log^3(1/\eps))}$ such that the following holds:
Let $f:\bn \to \bits$ be an LTF and let $0< \eps, \delta < 1/2$. Write $\mychows_f$ for the Chow vector of $f$
and assume that $\vec{\alpha} \in \R^{n+1}$ is a vector satisfying $\| \vec{\alpha} - \mychows_f \| \leq \chowallow(\eps)$. Then, there
is an algorithm $\alg$ with the following property: Given as input $\vec{\alpha}$, $\eps$ and $\delta$,
$\alg$ performs $ \tilde{O}(n^2) \cdot \poly(1/\chowallow(\eps)) \cdot \log(1/\delta)$ bit operations and outputs
the (weights-based) representation of an LTF $f^{\ast}$ which with probability at least $1-\delta$ satisfies
$\dist(f,f^{\ast}) \leq \eps$.
\end{theorem}

%We are now ready to prove Theorem~\ref{thm:main-formal}.

\begin{proof}[Proof of Theorem~\ref{thm:main-formal}]
Suppose that we are given a vector $\vec{\alpha} \in \R^{n+1}$ that satisfies $\Delta : = \| \vec{\alpha} - \mychows_f \| \leq \chowallow(\eps)$, where
$f$ is the unknown LTF to be learned. To construct the desired $f^{\ast}$, we run algorithm {\tt ChowReconstruct}  (from Theorem~\ref{thm:alg}) on input
$\vec{\alpha}$. The algorithm runs in time $\poly(1/\Delta) \cdot \tilde{O}(n^2) \cdot \log(1/\delta)$ and
outputs an LBF $g$ such that with probability at least $1-\delta$ we have
$\dchow(f,g) \leq 6 \Delta \leq 6 \chowallow(\eps)$.
(We can set the constants appropriately in the definition of the function $\chowallow(\eps)$ above,
so that the quantity on the RHS of the latter relation is smaller than the ``quasi-polynomial'' quantity
we need in the main structural theorem, so that the conclusion is ``$\dist(f,g) \leq \eps/2$''.) By Theorem~\ref{thm:dchow-vs-dh} we get that with probability at least $1-\delta$  we have $\dist(f,g) \leq \eps/2$.  Writing the LBF $g$ as $g(x) = P_1(v_0+\littlesum_{i=1}^n v_i x_i)$, we now claim that
$f^\ast(x) = \sign(v_0+\littlesum_{i=1}^n v_i x_i)$ has $\dist(f,f^\ast) \leq \eps.$  This is simply
because for each input $x \in \bn$, the contribution that $x$ makes to to $\dist(f,f^\ast)$ is
at most twice the contribution $x$ makes to $\dist(f,g).$ This completes the proof of Theorem~\ref{thm:main-formal}.
\end{proof}

\noindent As a simple corollary, we obtain Theorem~\ref{thm:lowwt}.

\begin{proof}[Proof of Theorem~\ref{thm:lowwt}]
Let $f:\bn \to \bits$ be an arbitrary LTF. We apply Theorem~\ref{thm:main-formal} above, for $\delta = 1/3$, and consider the LTF $f^{\ast}$ produced by the above proof.
Note that the weights $v_i$ defining $f^{\ast}$ are identical to the weights of the LBF $g$ output by  the algorithm {\tt ChowReconstruct}.
It follows from Theorem~\ref{thm:alg} that these weights are integers that satisfy $\littlesum_{i=1}^n v_i^2  = O(n \cdot \Delta^{-6})$, where
$\Delta = \Omega ( \chowallow(\eps) )$, and the proof is complete.
\end{proof}

\noindent As pointed out in Section~\ref{ssec:results} our algorithm runs in $\poly(n/\eps)$ time for LTFs whose integer weight is at most $\poly(n)$.
Formally, we have:

\begin{theorem} \label{thm:small-weights}
Let $f = \sign(\littlesum_{i=1}^n w_i x_i  - \theta)$ be an LTF with integer weights $w_i$ such that $W \eqdef \littlesum_{i=1}^n |w_i| = \poly(n)$.
Fix $0< \eps, \delta < 1/2$. Write $\mychows_f$ for the Chow vector of $f$
and assume that $\vec{\alpha} \in \R^{n+1}$ is a vector satisfying $\| \vec{\alpha} - \mychows_f \| \leq \eps/(12W)$. Then, there
is an algorithm $\alg'$ with the following property: Given as input $\vec{\alpha}$, $\eps$ and $\delta$,
$\alg'$ performs $ \poly(n/\eps) \cdot \log(1/\delta)$ bit operations and outputs
the (weights-based) representation of an LTF $f^{\ast}$ which with probability at least $1-\delta$ satisfies
$\dist(f,f^{\ast}) \leq \eps$.
\end{theorem}

\begin{proof}
As stated before, both the algorithm and proof of the above theorem are identical to the ones in Theorem~\ref{thm:main-formal}. The details follow.

Given a vector $\vec{\alpha} \in \R^{n+1}$ satisfying $\Delta : = \| \vec{\alpha} - \mychows_f \| \leq \eps/(12W)$, where
$f$ is the unknown LTF, we run algorithm {\tt ChowReconstruct}  on input
$\vec{\alpha}$. The algorithm runs in time $\poly(1/\Delta) \cdot \tilde{O}(n^2) \cdot \log(1/\delta)$,
which is $\poly(n/\eps) \cdot \log(1/\delta)$ by our assumption on $W$, and
outputs an LBF $g$ such that with probability at least $1-\delta$,
$\dchow(f,g) \leq 6 \Delta \leq \eps/(2W)$.
At this point, we need to apply the following simple structural result of~\cite{BDJ+:98}:

\begin{fact} \label{fact:small-weights}
Let $f = \sign(\littlesum_{i=1}^n w_i x_i  - \theta)$ be an LTF with integer weights $w_i$, where $W \eqdef \littlesum_{i=1}^n |w_i|$, and
$g:\bn \to [-1,1]$ be an arbitrary bounded function.
Fix $0< \eps < 1/2$. If $\dchow(f,g) \leq \eps/W$, then $\dist(f,g) \leq \eps.$
\end{fact}

The above fact implies that, with probability at least $1-\delta$,  the LBF $g$ output by the algorithm satisfies
$\dist(f,g) \leq \eps/2$.  If $g(x) = P_1(v_0+\littlesum_{i=1}^n v_i x_i)$, we similarly have that the LTF
$f^\ast(x) = \sign(v_0+\littlesum_{i=1}^n v_i x_i)$ has $\dist(f,f^\ast) \leq \eps.$ This completes the proof.
\end{proof}

%For the sake of completeness, we provide the details.
%Given a vector $\vec{\alpha} \in \R^{n+1}$ satisfying $\Delta : = \| \vec{\alpha} - \mychows_f \| \leq \eps/(12W)$, where
%$f$ is the unknown LTF, we run algorithm {\tt ChowReconstruct}  on input
%$\vec{\alpha}$. The algorithm runs in time $\poly(1/\Delta) \cdot \tilde{O}(n^2) \cdot \log(1/\delta)$,
%which is $\poly(n/\eps) \cdot \log(1/\delta)$ by our assumption on $W$, and
%outputs an LBF $g$ such that with probability at least $1-\delta$,
%$\dchow(f,g) \leq 6 \Delta \leq \eps/(2W)$.
%At this point, we need to apply the following simple structural result of~\cite{BDJ+:98}:

%\begin{fact} \label{fact:small-weights}
%Let $f = \sign(\littlesum_{i=1}^n w_i x_i  - \theta)$ be an LTF with integer weights $w_i$, where $W \eqdef \littlesum_{i=1}^n |w_i|$, and
%$g:\bn \to [-1,1]$ be an arbitrary bounded function.
%Fix $0< \eps < 1/2$. If $\dchow(f,g) \leq \eps/W$, then $\dist(f,g) \leq \eps.$
%\end{fact}
%
%The above fact implies that, with probability at least $1-\delta$,  the LBF $g$ output by the algorithm satisfies
%$\dist(f,g) \leq \eps/2$.  If $g(x) = P_1(v_0+\littlesum_{i=1}^n v_i x_i)$, we similarly have that the LTF
%$f^\ast(x) = \sign(v_0+\littlesum_{i=1}^n v_i x_i)$ has $\dist(f,f^\ast) \leq \eps.$ This completes the proof.

\subsection{Near-optimality of Theorem~\ref{thm:dchow-vs-dh}.}
\label{sec:nearoptchow}

Theorem~\ref{thm:dchow-vs-dh} says that if $f$ is an LTF and $g: \{-1,1\}^n \to [-1,1]$ satisfy
$\dchow(f,g) \le \epsilon$ then $\dist(f,g) \leq 2^{-\Omega ( \sqrt[3]{\log(1/\epsilon)} )}$.
It is natural to wonder whether the conclusion can be strengthened to ``$\dist(f,g)
\leq \eps^c$'' where $c>0$ is some absolute constant.  Here we observe that no conclusion of the
form ``$\dist(f,g) \leq 2^{-\gamma(\eps)}$'' is possible for any function
$\gamma(\eps) = \omega(\log(1/\eps)/\log\log(1/\eps))$.

To see this, fix $\gamma$ to be any function such that $$\gamma(\eps) = \omega(\log(1/\eps)/\log\log(1/\eps)).$$
If there were a stronger version of Theorem~\ref{thm:dchow-vs-dh} in which the conclusion is
``then $\dist(f,g) \leq 2^{-\gamma(\eps)}$,'' the arguments of Section~\ref{ssec:main-proofs} would give that
for any LTF $f$, there is an LTF $f' = \sign(v \cdot x - \nu)$ such that $\Pr[f(x) \neq f'(x)] \leq \eps$,
where each $v_i \in \Z$ satisfies $|v_i| \leq \poly(n) \cdot (1/\eps)^{o(\log \log(1/\eps))}.$  Taking $\eps = 1/2^{n+1}$, this
tells us that $f'$ must agree with $f$ on every point in $\{-1,1\}^n$, and each integer weight in
the representation $\sign(v \cdot x - \nu)$ is at most $2^{o(n \log n)}$.  But choosing $f$
to be H\aa stad's function from \cite{Hastad:94}, this is a contradiction, since any integer representation
of that function must have every $|v_i| \geq 2^{\Omega(n \log n)}$.

\section{Applications to learning theory}
\label{sec:learn}

In this section we show that our approach yields a range of interesting algorithmic applications in learning theory.

\subsection{Learning threshold functions in the 1-RFA model.} \label{ssec:rfa}
Ben-David and Dichterman~\cite{BenDavidDichterman:98}  introduced the ``Restricted Focus of Attention'' (RFA) learning framework
to model the phenomenon (common in the real world) of a learner having incomplete access to examples.
We focus here on the uniform-distribution ``$1$-RFA'' model.   In this setting
each time the learner is to receive a labeled
example, it first specifies an index $i \in [n]$; then an $n$-bit
string $x$ is drawn from the uniform distribution over $\bn$ and the learner is
given $(x_i, f(x))$.  So for each labeled example, the learner is only shown the $i$-th bit of the
example along with the label.

Birkendorf et al.~\cite{BDJ+:98} asked whether LTFs can be learned in the uniform
distribution $1$-RFA model, and showed that a sample of $O(n \cdot W^2 \cdot \log({\frac n
\delta})/\eps^2)$ many examples is information-theoretically
sufficient for learning an unknown threshold function with integer
weights $w_i$ that satisfy $\littlesum_i |w_i| \leq W.$
The results of Goldberg~\cite{Goldberg:06b} and Servedio~\cite{Servedio:07cc} show that samples of size
$(n/\eps)^{O(\log (n/\eps) \log(1/\eps))}$ and
$\poly(n)\cdot 2^{\tilde{O}(1/\eps^2)}$ respectively are information-theoretically sufficient for
learning an arbitrary LTF to accuracy $\eps$, but none of these earlier results gave a computationally efficient algorithm.
\cite{OS11:chow} gave the first algorithm for this problem; as a consequence of their result
for the Chow Parameters Problem, they gave an algorithm which learns LTFs to accuracy
$\eps$ and confidence $1 - \delta$ in the uniform distribution $1$-RFA
model, running in $2^{2^{\tilde{O}(1/\eps^2)}} \cdot n^2 \cdot
\log n \cdot \log({\frac n \delta})$ bit operations. As a direct consequence of Theorem~\ref{thm:main},
we obtain a much more time efficient learning algorithm for this learning task.

\begin{theorem} \label{thm:learnrfa} There is
an algorithm which performs
$\tilde{O}(n^2) \cdot (1/\eps)^{O(\log^2(1/\eps))} \cdot \log({\frac 1 \delta})$ bit-operations and properly learns LTFs to accuracy
$\eps$ and confidence $1 - \delta$ in the uniform distribution $1$-RFA model.
\end{theorem}

\subsection{Agnostic-type learning.} \label{ssec:agnostic}

In this section we show that a variant of our main algorithm gives a very fast ``agnostic-type'' algorithm for learning LTFs under the uniform distribution.

Let us briefly review the uniform distribution agnostic learning model~\cite{KSS:94} in our context. Let $f:\bn \to \bits$ be
an arbitrary boolean function. We write $\opt = \dist (f, \mathcal{H})  \eqdef  \min_{h \in \mathcal{H}} \Pr_x [h(x) \neq f(x)]$, where $\mathcal{H}$ denotes
the class of LTFs. A uniform distribution agnostic learning algorithm is given uniform random examples labeled according to an arbitrary $f$ and outputs a hypothesis
$h$ satisfying $\dist(h, f) \leq \opt + \eps.$

The only efficient algorithm for learning LTFs in this model~\cite{KKM+:05} is non-proper and runs in time $n^{\poly(1/\eps)}$.
This motivates the design of more efficient algorithms with potentially relaxed guarantees. \cite{OS11:chow} give an ``agnostic-type'' algorithm,
 that guarantees $\dist(h, f) \leq \opt^{\Omega(1)} + \eps$ and runs in time $\poly(n) \cdot 2^{\poly(1/\eps)}$. In contrast, we give an algorithm
that is significantly more efficient, but has a relaxed error guarantee.

\begin{theorem} \label{thm:agnostic}
There is an algorithm $\mathcal{B}$ with the following performance guarantee: Let $f$ be any Boolean function and let $\opt = \dist (f, \mathcal{H}).$
Given $0< \eps, \delta < 1/2$ and access to independent uniform examples $(x, f(x))$, algorithm $\mathcal{B}$ outputs the (weights-based) representation
of an LTF $f^{\ast}$ which with probability $1-\delta$ satisfies $\dist(f^{\ast},f) \leq 2^{-\Omega ( \sqrt[3]{\log(1/\opt)}) }+\eps$.
The algorithm performs $\tilde{O}(n^2) \cdot (1/\eps)^{O(\log^2(1/\eps))}  \cdot  \log(1/\delta)$ bit operations.
\end{theorem}

\begin{proof}
We describe the algorithm $\mathcal{B}$ in tandem with a proof of correctness.
We start by estimating each Chow parameter of $f$ (using the random labeled examples) to accuracy $O(\kappa(\eps)/\sqrt{n})$;
we thus compute  a vector $\vec{\alpha} \in \R^{n+1}$ that satisfies $\Delta : = \| \vec{\alpha} - \mychows_f \| \leq \chowallow(\eps)$.
We then run algorithm {\tt ChowReconstruct}  (from Theorem~\ref{thm:alg}) on input
$\vec{\alpha}$. The algorithm runs in time $\poly(1/\Delta) \cdot \tilde{O}(n^2) \cdot \log(1/\delta)$ and
outputs an LBF $g$ such that with probability at least $1-\delta$ we have
$\dchow(f,g) \leq 6 \Delta \leq 6 \chowallow(\eps)$.
By assumption, there exists an LTF $h^{\ast}$ such that $\dist (h^{\ast}, f) \leq \opt$.
By Fact~\ref{fact:distances} we get $\dchow(h^{\ast}, f) \leq 2\sqrt{\opt}$. An application of the triangle inequality now gives
$\dchow(g, h^{\ast}) \leq 2\sqrt{\opt} + 4 \chowallow(\eps)$.
By Theorem~\ref{thm:dchow-vs-dh}, we thus obtain $\dist(g, h^{\ast}) \leq 2^{-\Omega (\sqrt[3]{\log(1/\opt)})} + \eps/2$.
Writing the LBF $g$ as $g(x) = P_1(v_0+\littlesum_{i=1}^n v_i x_i)$, we similarly have that
$f^\ast(x) = \sign(v_0+\littlesum_{i=1}^n v_i x_i)$ has $\dist(f,f^\ast) \leq 2^{-\Omega (\sqrt[3]{\log(1/\opt)})}+ \eps.$
It is easy to see that the running time is dominated by the second step
and the proof of Theorem~\ref{thm:agnostic} is complete.
\end{proof}

\section{Conclusions and Open Problems} \label{sec:concl}
The problem of reconstructing a linear threshold function (exactly or approximately) from  (exact or approximate values of)
its degree-$0$ and degree-$1$ Fourier coefficients  arises in various contexts and has been considered by researchers in electrical engineering,
game theory, social choice and learning. In this paper, we gave an algorithm that reconstructs an $\eps$-approximate LTF (in Hamming distance)
and runs in time $\tilde{O}(n^2) \cdot (1/\eps)^{O(\log^2(1/\eps))}$, improving the only previous provably efficient algorithm~\cite{OS11:chow}
by nearly two exponentials (as a function of $\eps$). Our algorithm yields the existence of nearly-optimal integer weight approximations
for LTFs and gives significantly faster algorithms for several problems in learning theory.

We now list some interesting open problems:
\begin{itemize}

\item What is the complexity of the exact Chow parameters problem? The problem is easily seen to lie in $NP^{PP}$, and
we are not aware of a better upper bound. We believe that the problem is intractable; in fact, we conjecture it is $PP$-hard.

\item Is there an FPTAS for the problem, i.e. an algorithm running in $\poly(n/\eps)$ time? (Note that this would be best possible, assuming that the
exact problem is intractable; in this sense our attained upper bound is close to optimal.) We believe so; in fact, we showed this is the case for $\poly(n)$ integer weight LTFs.
(Note however that the arguments of Section~\ref{sec:nearoptchow} imply that our algorithm does {\em not} run in $\poly(n/\eps)$ time for general LTFs, and indeed imply that no algorithm that outputs a $\poly(n/\eps)$-weight LTF can succeed for this problem.)
%It may be the case that significantly new ideas are required to prove this for general LTFs.

\item What is the optimal bound in Theorem~\ref{thm:dchow-vs-dh}? Any improvement would yield an improved running time for our algorithm.

\item Our algorithmic approach is quite general. As was shown in \cite{Feldman:12colt}, this approach can also be used to learn small-weight low-degree PTFs. In addition, essentially the same algorithm was more recently used~\cite{DDS:12icalp} to solve a problem in social choice theory. Are there any other applications of our boosting-based approach?

\item Does our structural result generalize to degree-$d$ PTFs? A natural generalization of ChowÕs theorem holds in this setting; more precisely, Bruck~\cite{Bruck:90} has shown that the Fourier coefÞcients of degree at most $d$ uniquely specify any degree-$d$ PTF within the space of all Boolean or even bounded functions.
Is there a ``robust version'' of Bruck's theorem? We consider this to be a challenging open problem.
(Note that our algorithmic machinery generalizes straightforwardly to this setting, hence a robust such result would immediately yield an efficient algorithm in this generalized setting.)

\end{itemize}

\bibliography{allrefs}

\bibliographystyle{alpha}

\appendix

\section{Near-Optimality of Lemma~\ref{lem:1}} \label{ap:optimalbeta}

The following lemma shows that in any statement like Lemma~\ref{lem:1} in which the
hyperplane $\h'$ passes through \emph{all} the points in $S$, the distance bound on $\beta$
can be no larger than $n^{-1/2}$ as a function of $n$.  This implies that the result obtained
by taking $\kappa = 1/2^{n+1}$ in Lemma~\ref{lem:1}, which gives a distance bound of $n^{-(1/2 + o(1))}$ as
a function of $n$, is optimal up to the $o(1)$ in the exponent.

\begin{lemma}
\label{lem:nearopt}
Fix $\eps > 8n^{-1/2}.$
There is a hyperplane $\h \in \mathbb{R}^n$ and a set $S \subseteq \{-1,1\}^n$ such that $|S| \ge {\frac \eps 8} 2^n$ and the following properties both hold:

\begin{itemize}

\item For every $x \in S$ we have  $d(x,\h) \le 2 \eps n^{-1/2}$; and
\item There is no hyperplane $\h'$ which passes through all the points in $S$.
\end{itemize}
\end{lemma}
\begin{proof}
Without loss of generality, let us assume $K = 4/\epsilon^2$ is an even integer; note that by assumption $K < n/2.$  Now let us define the hyperplane
$\h$  by
$$
\h = \left\{x \in \mathbb{R}^n : (x_1 + \ldots +x_K)  + \frac{2(x_{K+1} + \ldots +x_n)}{(n-K)}  = 0\right\}
$$
Let us define $S = \{x \in \{-1,1\}^n  : d(x,\h) \le 4/ \sqrt{K (n-K)} \}$. It is easy to verify that
every $x \in S$ indeed satisfies $d(x,\h) \leq 2 \eps n^{-1/2}$ as claimed. Next, let us define $A$ as follows:
$$
A = \{x \in \{-1,1\}^n  : x_1 + \ldots + x_K=0$$ and  $$|x_{K+1} + \ldots +x_n | \le 2 \sqrt{n-K}\}.$$
It is easy to observe that $A \subseteq S$.  Also, we have
 $$\Pr_{x_1,\ldots,x_K} [x_1+\ldots +x_K=0] \ge (2\sqrt{K})^{-1}$$ and
 $$\Pr_{x_{K+1} , \ldots, x_n} [ |x_{K+1} + \ldots +x_n | \le 2 \sqrt{n-K} ] \ge 1/2.$$
Hence we have that $|S| \ge \epsilon 2^n/8.$
We also observe that the point $z \in \{-1,1\}^n$ defined as
\begin{equation}
\label{eq:z}
z:= (1,1,\underbrace{1,-1,\ldots, 1,-1}_{K-2}, -1,\ldots,-1)
\end{equation}
(whose first two coordinates are 1, next $K-2$ coordinates
alternate between $1$ and $-1$, and final $n-K$ coordinates
are $-1$)
lies on $\h$ and  hence $z \in S$.

We next claim that the dimension of the affine span of the points in $A \cup z$ is $n$. This obviously implies that there is no hyperplane which passes through all points in $A \cup z$, and hence no hyperplane which passes through all points in $S$.  Thus to prove the
lemma it remains only to prove the following claim:

\begin{claim}
The dimension of the affine span of the elements of $A \cup z$ is $n$.
\end{claim}
To prove the claim, we observe that if we let $Y$ denote
the affine span of elements in $A \cup z$ and $Y'$ denote
the linear space underlying $Y$,
then it suffices to show that the dimension of $Y'$ is $n$.  Each element
of $Y'$ is obtained as the difference of two elements in $Y$.

First, let $y \in \{-1,1\}^n$ be such that $$\sum_{i\le K} y_i  = \sum_{K+1 \le i \le n} y_i =0.$$  Let $y^{\oplus i} \in \{-1,1\}^n$
be obtained from $y$ by flipping the $i$-th bit.  For each $i \in
\{K + 1,\dots, n\}$ we have that $y$ and $y^{\oplus i}$ are both in $A$,
so subtracting the two elements, we get that the basis vector $e_i$ belongs
to $Y'$ for each $i \in \{K+1,\dots,n\}.$

Next, let $i \neq j \leq K$ be positions such that $y_i =1$ and
$y_j =-1$.  Let $y^{ij}$ denote the vector which is the same as $y$
except that the signs are flipped at coordinates $i$ and $j$.
Since $y^{ij}$ belongs to $A$, by subtracting $y$ from $y^{ij}$ we get that
for every vector $e_{ij}$ ($i \not = j \le K$) which has 1 in coordinate
$i$, $-1$ in coordinate $j$, and 0 elsewhere, the vector
$e_{ij}$ belongs to $Y'$.

The previous two paragraphs are easily seen to imply that the linear
space $Y'$ contains all vectors $x \in \R^n$ that satisfy the condition
$x_1 + \cdots + x_K = 0.$  Thus to show that the dimension of $Y'$
is $n$, it suffices to exhibit any vector in $Y'$ that
does not satisfy this condition.  But it is
easy to see that the vector $y - z$ (where $z$ is defined in (\ref{eq:z}))
is such a vector.  This concludes the proof of the claim and of
Lemma~\ref{lem:nearopt}.
\end{proof}

\ignore{
Finally, define $z' \in \{-1,1\}^n$ as the vector
$$
z':= (-1,1,1,-1,\ldots,1,-1)
$$
whose coordinates alternate between $-1$ and $1.$ {\huge UP TO HERE} Subtracting $z'$ from $z$,
we get a vector $z'' \in Y'$ which is of the following form : $z''$ has
$1$ in the first position, $0$ in the next $K-1$ positions and $1$ in the odd numbered positions of the next $n-K$ positions. Subtracting vectors of the form $e_i$ (for odd $i$'s in the interval $(K,n-K]$), we get that $e_1=(1,0,\ldots, 0)$ is in $Y'$. It is now easy to see that the dimension of the linear span of $e_1$, $e_i$ (for $K<i \le n$) and $e^{ij}$ for $i \not = j \le K$ is $n$ proving our claim.
}

\section{Useful variants of Goldberg's theorems} \label{sec:usefulgoldberg}

For technical reasons we require an extension of
Theorem~\ref{thm:goldberg-thm3} (Theorem~3 of \cite{Goldberg:06b}) which
roughly speaking is as follows: the hypothesis is that not only does the
set $S \subset \{-1,1\}^n$ lie close to hyperplane $\h$ but so also does
a (small) set $R$ of points in $\{0,1\}^n$; and the conclusion is that
not only does ``almost all'' of $S$ (the subset $S^*$) lie on $\h'$
\emph{but so also does all of $R$}. To obtain this extension we need a
corresponding extension of an earlier result of Goldberg (Theorem~2 of
\cite{Goldberg:06b}), which he uses to prove his Theorem~3; similar to
our extension of Theorem~\ref{thm:goldberg-thm3} our extension of
Theorem~2 of \cite{Goldberg:06b} deals with points from both
$\{-1,1\}^n$ and $\{0,1\}^n.$ The simplest approach we have found to
obtain our desired extension of Theorem~2 of \cite{Goldberg:06b} uses
the ``Zeroth Inverse Theorem'' of Tao and Vu \cite{TaoVu:annals}.  We
begin with a useful definition from their paper:

%In order to state the results here, the following definition from Tao-Vu \cite{TaoVu:annals} will be very helpful.

\begin{definition} Given a vector ${w}=(w_1,\dots,w_k)$ of real
values, the \emph{cube $S({w})$} is the subset of $\R$ defined as
\footnote{In \cite{TaoVu:annals} the cube is defined only allowing
$\eps_i \in \{-1,1\}$ but this is a typographical error; their proof
uses the $\eps_i \in \{-1,0,1\}$ version that we state.} $$
S({w}) = \left\{ \sum_{i=1}^k \epsilon_i w_i : (\epsilon_1,
\ldots, \epsilon_n) \in \{-1,0,1\}^n\right\}. $$ \end{definition}

The ``Zeroth Inverse Theorem'' of \cite{TaoVu:annals} is as follows:

\begin{theorem}\label{thm:taovu}
Suppose $w \in \mathbb{R}^n$, $d \in \mathbb{N}$ and  $\theta \in \mathbb{R}$
satisfy $\Pr_{x \in \{-1,1\}^n} [ w \cdot x =\theta ] > 2^{-d-1}$. Then
there exists a $d$-element subset $A=\{i_1,\dots,i_d\} \subset [n]$
such that for ${v} = (w_{i_1},\dots,w_{i_d})$ we have
 $\{w_1, \ldots, w_n\} \subseteq S({v})$.
\end{theorem}

For  convenience of the reader, we include the proof here.
\begin{proof}[Proof of Theorem~\ref{thm:taovu}]
Towards a contradiction, assume that there is no ${v} = (w_{i_1},\dots,w_{i_d})$ such that
 $\{w_1,$ $\ldots,$ $w_n\} \subseteq S({v})$. Then an obvious greedy argument shows that there are
 distinct integers $i_1, \ldots, i_{d+1} \in [n]$ such that $w_{i_1}, \ldots, w_{i_{d+1}}$ is \emph{dissociated},
i.e. there does not exist $j \in [n]$ and $\epsilon_{i} \in \{-1,0,1\}$ such that
 $ w_j = \littlesum_{i \not =j} \epsilon_i w_i$.

Let $v =(w_{i_1}, \ldots, w_{i_{d+1}})$. By an averaging argument, it is easy to see that
if $\Pr_{x \in \{-1,1\}^n} [ w \cdot x =\theta ] > 2^{-d-1}$, then $\exists \nu \in \mathbb{R}$ such that
$\Pr_{x \in \{-1,1\}^{d+1}} [ v \cdot x =\nu ] > 2^{-d-1}$. By the pigeon hole principle, this means that there exist $x,y \in\{-1,1\}^{d+1}$ such that $x\not =y$ and $v \cdot ((x-y)/2)=0$. Since entries of $(x-y)/2$ are in $\{-1,0,1\}$, and not all the entries in $(x-y)/2$ are zero, this means that $v$ is not dissociated resulting in a contradiction.
\end{proof}

Armed with this result, we now prove the extension of Goldberg's
Theorem~2 that we will need later:

\begin{theorem}\label{thm:goldberg2}
Let $w \in \mathbb{R}^n$ have $\Vert w \Vert_2=1$ and let
$\theta \in \mathbb{R}$ be such that $\Pr_{x \in \{-1,1\}^n} [ w\cdot x =
\theta] = \alpha$. Let $\h$ denote the hyperplane
$\h = \{ x  \in \mathbb{R}^n \mid  w \cdot x = \theta \}$.
Suppose that $\span(
\h \cap( \{-1,1\}^n \cup \{0,1\}^n)) = \h$,
i.e. the affine span of the points in
$\{-1,1\}^n \cup \{0,1\}^n$ that lie on $\h$ is $\h$.
Then all entries of $w$ are integer multiples of $f(n,\alpha)^{-1}$, where
$$
f(n,\alpha) \leq (2n)^{\lfloor \log(1/\alpha) \rfloor + 3/2} \cdot (\lfloor \log(1/\alpha) \rfloor)!
$$
\end{theorem}
\begin{proof}
We first observe that $w \cdot( x-y) =0$ for any two points $x,y$ that both lie on $\h.$
Consider the system of homogeneous linear equations in variables $w'_1,\dots,w'_n$ defined by
\begin{equation} \label{eq:system}
w' \cdot (x-y) =0 \quad \quad \mbox{for all~}x,y \in \h \cap( \{-1,1\}^n \cup \{0,1\}^n).
\end{equation}
Since $\span(\h \cap( \{-1,1\}^n \cup
\{0,1\}^n))$ is by assumption the entire hyperplane $\h$, the system (\ref{eq:system}) must
have rank $n-1$; in other words, every solution $w'$ that satisfies (\ref{eq:system}) must be some rescaling
$w' = cw$ of the vector $w$ defining $\h$.

Let $A$ denote a subset of $n-1$ of the equations comprising (\ref{eq:system}) which has rank $n-1$
(so any solution to $A$ must be a vector $w' = cw$ as described above).  We note that each coefficient
in each equation of $A$ lies in $\{-2,-1,0,1,2\}.$  Let us define $d = \lfloor \log(1/\alpha) \rfloor +1$. By Theorem~\ref{thm:taovu},
there is some $w_{i_1}, \ldots, w_{i_{d'}}$ with $d' \leq d$ such that for ${v} \eqdef
(w_{i_1},\ldots, w_{i_{d'}})$, we have  $\{w_1, \ldots, w_n\} \subseteq
S({v})$; in other words, for all $j \in [n]$ we have
$w_j = \sum_{\ell =1}^{d'} \epsilon_{\ell,j} w_{i_\ell}$ where each $\eps_{\ell,j}$
belongs to $\{-1,0,1\}.$
% there is a set $A \subset [n]$ of size at
%most $d=\lfloor \log(1/\alpha) \rfloor$ such that for ${v} =
%(w_i)_{i \in A}$, we have $\{w_1, \ldots, w_n\} \subseteq
%S({v})$.
% $\h \cap( \{-1,1\}^n \cup \{0,1\}^n)$ has rank $n-1$.  In other
%words, the vector $w$ is determined up to the scaling factor.
%
%Let us now choose $A$ to be a subset of these equations which have rank
%$n-1$. We observe that the coefficients in the equations are in the set
%$\{-2,-1,0,1,2\}$. We next use Theorem~\ref{thm:taovu} and infer that
%there is some $w_{i_1}, \ldots, w_{i_d}$ such that if ${v}=
%\{w_{i_1},\ldots, w_{i_d} \}$ such that $\forall j \in [n]$, $w_j =
%\sum_{\ell =1}^d \epsilon_{\ell,j} w_{i_\ell}$.
Substituting these relations into the system $A$, we  get a new system of homogenous
linear equations, of rank $d'-1$, in the variables $w'_{i_1},
\ldots, w'_{i_{d'}}$, where all coefficients of all variables in all equations of the system are
 integers of magnitude at most $2n.$

Let $M$ denote a subset of $d'-1$ equations from this new system which has rank $d'-1.$
In other words, viewing $M$ as a $d' \times (d' - 1)$
matrix, we have the equation
$M \cdot {v}^T =0$ where all entries in the matrix $M$
are integers in $[-2n,2n]$. Note that at least one of the values
$w_{i_1}, \ldots , w_{i_{d'}}$ is non-zero (for if all of them were 0, then since
 $\{w_1, \ldots, w_n\} \subseteq S({v})$ it would have to be the case that $w_1 = \cdots = w_n = 0.$).
Without loss of generality we may suppose that $w_{i_1}$ has the largest magnitude among
$w_{i_1},\dots,w_{i_{d'}}$.  We now fix the scaling constant $c$, where $w'=cw$, to be such that
$w'_{i_1} = 1.$  Rearranging the system $M (c {v})^T = M (1,w'_{i_2},\dots,w'_{i_{d'}})^T = 0$,
we get a new system of $d'-1$ linear equations $M' (w'_{i_2},\dots,w'_{i_{d'}})^T = b$
where $M'$ is a $(d'-1) \times (d'-1)$ matrix whose entries are integers in $[-2n,2n]$
and $b$ is a vector whose entries are integers in $[-2n,2n].$

%\bigskip \bigskip
%
%Let us choose the largest
%of $w_{i_1},\ldots, w_{i_d}$ in magnitude and we set it to $1$. Without
%loss of generality, let it be $w_{i_1}$. In other words, we set
%$w_{i_1}=1$. We then have a new system of linear equation $M' \cdot
%{w''} = b$ where ${w''} = (w_{i_2}, \ldots, w_{i_d})$ and
%the entries of $b$ are integers in the interval $[-2n,2n]$.
%

We now use Cramer's rule to solve the system $$M' (w'_{i_2},\dots,w'_{i_{d'}})^T = b.$$
This gives us that
$w'_{i_j} = \det(M'_j)/\det(M')$ where $M'_j$ is the matrix obtained by
replacing the $j^{th}$ column of $M'$ by $b$. So each
$w'_{i_j}$ is an integer multiple of $1/\det(M')$ and is bounded by $1$
(by our earlier assumption about $w_{i_1}$ having the largest magnitude).
Since $\{w'_1, \ldots, w'_n \} \subseteq S({v})$,
we get that each value $w'_i$ is an integer multiple of $1/\det(M')$, and each
$|w'_i| \leq n.$  Finally, since  $M'$ is a $(d'-1)
\times (d'-1)$
matrix where every entry is an integer of magnitude at most $2n$, we have  that
$|\det(M') | \le (2n)^{d'-1} \cdot (d'-1) ! \leq (2n)^{d-1} \cdot (d-1)!$.  Moreover, the
$\ell_2$ norm of the vector $w'$ is bounded by $n^{3/2}$. So renormalizing (dividing by $c$)
to obtain the unit vector $w$ back from $w' = c w$,  we see that
every entry of $w$ is an integer multiple of $1/N$, where $N$ is a quantity at most
$(2n)^{d + 1/2}
\cdot d!$. Recalling that $d = \lfloor \log(1/\alpha) \rfloor +1$, the theorem is proved.
\end{proof}

We next  prove the extension of Theorem~3 from \cite{Goldberg:06b} that we require. The proof is almost identical to the proof in \cite{Goldberg:06b} except for the use of Theorem~\ref{thm:goldberg2} instead of Theorem~2 from \cite{Goldberg:06b} and a few other syntactic changes. For the sake of clarity and completeness, we give the complete proof here.

\begin{theorem} \label{thm:newgoldberg3}
Given any hyperplane $\h$ in $\mathbb{R}^n$ whose $\beta$-neighborhood contains a subset $S$ of vertices of $\{-1,1\}^n$ where $S=\alpha \cdot 2^n$, there exists a hyperplane which passes through all the points of $(\{-1,1\}^n \cup \{0,1\}^n)$ that are contained in the $\beta$-neighborhood of $\h$ provided that
$$
0 \le \beta \le \left( (2/\alpha) \cdot n^{5 + \lfloor \log(n/\alpha)\rfloor }   \cdot (2 + \lfloor \log(n/\alpha) \rfloor)! \right)^{-1}.
$$
\end{theorem}

Before giving the proof, we note that the hypothesis of our theorem is the same as the
hypothesis of Theorem~3 of \cite{Goldberg:06b}. The only difference in the conclusion is that while Goldberg proves that all points of $\{-1,1\}^n$ in the $\beta$-neighborhood of $\h$ lie on the new hyperplane, we prove this for all the points of  $(\{-1,1\}^n \cup \{0,1\}^n)$  in the $\beta$-neighborhood of $\h$.

\begin{proof}
Let $\h = \{x \mid w\cdot x - t =0\}$ with $\Vert w \Vert =1$.  Also, let $S = \{x \in \{-1,1\}^n \mid d(x,\h) \le \beta\}$ and $S' = \{x \in (\{-1,1\}^n \cup \{0,1\}^n)  \mid d(x,\h) \le \beta\}$.
For any $x \in S'$ we have that $w\cdot x  \in [t-\beta, t+\beta]$. Following \cite{Goldberg:06b} we create a new weight vector $w' \in \mathbb{R}^n$ by rounding each coordinate $w_i$ of $w$ to the nearest integer multiple of $\beta$ (rounding up in case of a tie).  Since every $x \in S'$ has entries from $\{-1,0,1\}$, we can deduce that for any $x \in S'$, we have
$$
t - \beta -n \beta/2 <
 w \cdot x - n \beta/2 < w' \cdot x < w \cdot x + n\beta/2 \le t + \beta + n \beta/2.
$$
Thus  for every $x \in S'$, the value $w' \cdot x$ lies in a semi-open interval of length $\beta (n+2)$;
moreover, since it only takes values which are integer multiples of $\beta$,
there are at most $n+2$ possible values that $w' \cdot x$ can take for $x \in S'.$
Since $S \subset S'$ and $|S| \ge \alpha 2^n$, there must be at least one value $t' \in (t - n \beta/2 - \beta , t + n \beta/2 + \beta]$ such that at least $\alpha2^n/(n+2)$ points in $S$ lie on the hyperplane $\h_1$ defined as $\h_1 = \{x : w' \cdot x = t'\}$.
We also let $A_1 = \span \{x \in S' : w' \cdot x =t'\}$. It is clear that $A_1 \subset \h_1$.  Also, since at least $\alpha2^n/(n+2)$ points of $\{-1,1\}^n$ lie on $A_1$, by Fact~\ref{fac:affine} we get that $\dim(A_1) \ge n - \log(n+2) -\log(1/\alpha)$.

It is easy to see that $\Vert w' - w \Vert \le \sqrt{n} \beta/2$, which implies that $\Vert w' \Vert \ge 1 - \sqrt{n} \beta/2$. Note that for any $x \in S'$ we have $|w' \cdot x -t'| \le (n+2) \beta$.  Recalling Fact~\ref{fac:hyper-dist}, we get that for any $x \in S'$ we have $d(x, \h_1) \le (\beta (n+2))/(1- \sqrt{n} \beta/2)$.  Since $\sqrt{n} \beta \ll 1$, we get that $d(x, \h_1) \le 2n \beta$ for every $x \in S'.$

At this point our plan for the rest of the proof of Theorem~\ref{thm:newgoldberg3} is as follows:  First we will construct a hyperplane $\h_k$ (by an inductive construction) such that $\span(\h_k \cap (\{-1,1\}^n \cup \{0,1\}^n)) = \h_k$, $A_1 \subseteq \h_k$, and all points in $S'$ are very close to $\h_k$ (say within
Euclidean distance $\gamma$).  Then we will apply Theorem~\ref{thm:goldberg2} to conclude that any point $\{-1,1\}^n \cup \{0,1\}^n$ which is not on $\h_k$ must have Euclidean distance at least some $\gamma'$ from $\h_k$. If $\gamma' > \gamma$ then  we can infer that every point in $S'$ lies on $\h_k$, which proves the theorem. We now describe the construction that gives $\h_k.$

If $\dim(A_1) = n-1$, then we let $k=1$ and stop the process, since as desired we have $\span(\h_k \cap (\{-1,1\}^n \cup \{0,1\}^n)) = \h_k$, $A_1 = H_k$, and $d(x,\h_k) \leq 2 n \beta$ for every $x \in S'.$
 Otherwise, by an inductive hypothesis, we may assume that for some $j \ge 1$ we have an affine space $A_j$ and a hyperplane $\h_j$ such that

 \begin{itemize}

 \item $A_1 \subseteq A_j \subsetneq \h_j$;
 \item $\dim(A_j) = \dim(A_1) + j-1$, and
 \item for all $x \in S'$ we have $d(x,\h_j) \le 2^{j} n \beta$.

 \end{itemize}

 Using this inductive hypothesis, we will construct an affine space $A_{j+1}$ and a hyperplane $\h_{j+1}$ such that $A_1 \subset A_{j+1} \subseteq \h_{j+1},$  $\dim(A_{j+1}) = \dim(A_1) + j$, and for all $x \in S'$ we have $$d(x,\h_{j+1}) \le 2^{j+1} n \beta.$$ If $A_{j+1} = \h_{j+1}$, we stop the process, else we continue.

We now describe the inductive construction. Since $A_j \subsetneq \h_j$, there must exist an affine
subspace $A'_j$ such that $A_j \subseteq A'_j \subsetneq \h_j$ and $\dim(A'_j) =n-2$. Let $x_j$ denote $\arg \max_{x \in S'} d(x, A'_j)$. (We assume that $ \max_{x \in S'} d(x, A'_j)>0$; if not, then choose $x_j$ to be an arbitrary point in $\{-1,1\}^n$ not lying on $A'_j$. In this case, the properties of the inductive construction will trivially hold.)  Define $\h_{j+1} = \span(A'_j \cup x_j)$.  It is clear that $\h_{j+1}$ is a hyperplane. We claim that
for $x \in S'$ we have
$$
d(x,\h_{j+1}) \le d(x,\h_j) + d(x_j,\h_j) \le 2^{j} n \beta + 2^j n \beta = 2^{j+1} n \beta.
$$
To see this, observe that without loss of generality we may assume that $\h_j$ passes through the origin and thus $A'_j$ is a linear subspace. Thus we have that $\Vert x_{\perp A'_j} \Vert  \le
\Vert (x_j)_{\perp A'_j} \Vert$ for all $x\in S'$, where for a point $z \in \R^n$ we write $z_{\perp A'_j}$ to denote the component of $x$ orthogonal to $A'_j$. Let $r = \Vert x_{\perp A'_j} \Vert$ and $r_1 = \Vert x_{j,\perp A'_j} \Vert$, where $r_1\ge r$.  Let $\theta$ denote the angle that $x_{\perp A'_j}$ makes with $\h_j$ and let $\phi$ denote the angle that $x_{\perp A'_j}$ makes with  $(x_j)_{\perp A'_j}$.  Then it is easy to see that $d(x,\h_{j+1}) = |r \cdot \sin(\theta -\phi)|$, $d(x,\h_{j}) = |r \cdot \sin(\theta)|$ and $d(x_j,\h_{j}) = |r_1 \cdot \sin(\phi)|$. Thus, we only need to check that if $r_1\ge r$, then $|r \cdot \sin(\theta -\phi)| \le |r \cdot \sin(\theta)| + |r_1 \cdot \sin(\phi)|$ which is straightforward to check.

Let $A_{j+1} = \span(A_j \cup x_j)$ and note that $A_1\subset A_{j+1} \subseteq \h_{j+1}$ and $\dim(A_{j+1}) = \dim(A_j) + 1 $. As shown above,  for all $x \in S'$ we have $d(x, \h_{j+1}) \le 2^{j+1} n \beta$. This completes the inductive construction.

Since $\dim(A_1) \ge n - \log(n+2) - \log(1/\alpha)$, the process must terminate for some $k \le \log(n+2) + \log(1/\alpha)$. When the process terminates, we have a hyperplane $\h_k$ satisfying the following properties:
\begin{itemize}
\item $\span(\h_k \cap (\{-1,1\}^n \cup \{0,1\}^n)) = \h_k$; and
\item $|\h_k \cap S| \ge \alpha2^n/(n+2)$; and
\item for all $x\in S'$ we have $d(x, \h_k) \le 2^k n \beta \le (1/\alpha) n(n+2) \beta$.
\end{itemize}
We can now apply Theorem~\ref{thm:goldberg2} to the hyperplane $\h_k$ to get that if $\h_k = \{x  \mid v \cdot x - \nu=0\}$ with $\Vert v \Vert =1$, then all the entries of $v$ are integral multiples of a quantity $E^{-1}$ where
$$
E \leq  (2n)^{\lfloor \log((n+2)/\alpha)\rfloor + 3/2}   \cdot (\lfloor \log((n+2)/\alpha) \rfloor)! .
$$
Consequently $v \cdot x$ is an integral multiple of $E^{-1}$ for every $x \in (\{-1,1\}^n \cup \{0,1\}^n)$.  Since there are points of $\{-1,1\}^n$ on $\h_k$, it must be the case that $\nu$ is also an integral multiple of $E$.  So if any $x \in  (\{-1,1\}^n \cup \{0,1\}^n)$ is such that $d(x, \h_k) < E$, then $d(x,\h_k)=0$
and hence $x$ actually lies on $\h_k$.  Now recall that for any $x\in S'$ we have $d(x,\h_k) \le (n/\alpha) (n+2) \beta$.  Our upper bound on $\beta$ from the theorem statement ensures that  $(n/\alpha) (n+2) \beta <E^{-1}$, and consequently every $x \in S'$ must lie on $\h_k$, proving the theorem.
\end{proof}

\end{document}